\documentclass[journal]{IEEEtran}

		\usepackage{amsmath,amsfonts}
		\usepackage{algorithmic}
		\usepackage{array}
		\usepackage[caption=false,font=normalsize,labelfont=sf,textfont=sf]{subfig}
		\usepackage{textcomp}
		\usepackage{stfloats}
		\usepackage{url}
		\usepackage{verbatim}
		\usepackage{graphicx}
		\def\BibTeX{{\rm B\kern-.05em{\sc i\kern-.025em b}\kern-.08em
		    T\kern-.1667em\lower.7ex\hbox{E}\kern-.125emX}}
		\usepackage{balance}

\usepackage{enumerate}
\usepackage{graphicx}
\usepackage{mathtools}
\usepackage{xspace}
\usepackage{graphicx}
\usepackage{subfig}
\usepackage[lined,linesnumbered,ruled,commentsnumbered]{algorithm2e} 
\usepackage{amsthm}
\usepackage{hyperref}
\usepackage{caption}
\usepackage{float}
\usepackage{dsfont}

\usepackage{tikz}
\usepackage{pgfplots}
\usepgfplotslibrary{patchplots} 
\usepgfplotslibrary{groupplots}
\pgfplotsset{compat=newest}
\usetikzlibrary{arrows,shapes,backgrounds,plotmarks,positioning}
\pgfplotsset{compat=newest}                         
\pgfplotsset{plot coordinates/math parser=false}

\newtheorem{theorem}{Theorem}

\newtheorem{corollary}{Corollary}
\newtheorem{proposition}{Proposition}

\newtheorem{remark}{Remark}

\newcommand*\dif{\mathop{}\!\mathrm{d}}

\newcommand{\X}{\mathcal{X}}
\newcommand{\Y}{\mathcal{Y}}

\newcommand{\Lap}{\text{Lap}}
\newcommand{\Gauss}{\text{Gauss}}
\newcommand{\Exp}{\text{Exp}}

\newcommand{\Real}{\mathds{R}}

\newcommand{\SetPair}{\mathds{S}}
\newcommand{\E}{\mathds{E}}

\newcommand{\Set}[1]{\{#1\}}

\newcommand{\supp}{\text{supp}}

\begin{document}
\title{$\alpha$-Wasserstein Mechanism for R\'{e}nyi Pufferfish Privacy}
\author{Ni Ding,~\IEEEmembership{Member,~IEEE}, Wenjin Yang,~\IEEEmembership{Student Member,~IEEE}, Zijian Zhang,~\IEEEmembership{Senior Member,~IEEE}
}

\markboth{Journal of \LaTeX\ Class Files,~Vol.~18, No.~9, September~2020}%
{$\alpha$-Wasserstein Mechanism for R\'{e}nyi Pufferfish Privacy}

\maketitle

\begin{abstract}
		This paper introduces the $\alpha$-Wasserstein mechanism for achieving R\'{e}nyi Pufferfish Privacy using Laplace and Gaussian noise. By leveraging H\"{o}lder's inequality, we demonstrate that the scale parameter of the Laplace mechanism can be calibrated via an upper bound on the $W_\alpha$ metric to satisfy $(\alpha, \epsilon)$-R\'{e}nyi Pufferfish Privacy for $\alpha \in (1, \infty]$. We show that at the limit $\alpha = \infty$, this framework recovers the established $W_\infty$ mechanism for $\epsilon$-pufferfish privacy. This result is subsequently extended to the exponential mechanism. Furthermore, we propose an $\alpha$-Wasserstein mechanism for Gaussian noise for $\alpha \in (1, \infty)$, demonstrating that it generalizes existing results within the Rényi Differential Privacy framework. Experimental evaluations reveal that our $\alpha$-Wasserstein mechanism significantly reduces noise power compared to the conventional $W_\infty$-based approach, with the Gaussian mechanism providing superior utility over the Laplace mechanism. Notably, the mechanisms derived in this work achieve exact $(\alpha, \epsilon)$-R\'{e}nyi Pufferfish Privacy without requiring additional relaxations, such as $\delta$-approximations.
\end{abstract}


\section{Introduction}
\IEEEPARstart{R}{ooted} in a rigorous mathematical framework of statistical indistinguishability, differential privacy provides a robust guarantee that the inclusion or exclusion of a single record remains probabilistically undetectable by bounding output variations within a privacy budget $\epsilon$ \cite{CalibNoiseDP, Dwork2006, Wasserman2010}. By ensuring that an adversary cannot reliably infer an individual's presence or specific contribution from observable outputs, differential privacy has emerged as the gold standard for privacy-preserving data analysis. Due to its formal security properties, differential privacy is now widely deployed across various domains, including official statistics \cite{Abowd2018USCensus}, machine learning \cite{Abadi2016DPSGD} and healthcare \cite{Mohammadi2025Medicine}. 

The Pufferfish framework extends the principles of differential privacy to scenarios where the original data, such as a query response, exhibits probabilistic dependence on a secret \cite{Pufferfish2012KiferConf,Pufferfish2014Kifer}. In this setting, the challenge lies in achieving statistical indistinguishability within the posterior data distribution following sanitization. To address this, the first noise calibration method was introduced by Song et al. (2017), who proposed setting the scale parameter $b$ of zero-mean Laplace noise according to the $\infty$-order Wasserstein metric to satisfy $\epsilon$-pufferfish privacy \cite{PufferfishWasserstein2017Song}.
However, computing the $\infty$-Wasserstein metric is complicated by its non-convex nature \cite{Champion2008InfW, DePascale2019InfW}. To resolve these computational difficulties, Ding (2022) introduced a $1$-order Wasserstein (Kantorovich) approach for both Laplace and Gaussian noise mechanisms \cite{Ding2022AISTATS}.

While strict privacy constraints often degrade data utility—a primary concern in differential privacy literature \cite{SoriaComas2017, Li2024}—one may resort to relaxations such as R\'{e}nyi measures. Similar to $(\epsilon, \delta)$-differential privacy, these relaxations bound the probability of a data breach within specified limits. Building on the principles of R\'{e}nyi Differential Privacy \cite{RDP2017}, $\epsilon$-pufferfish privacy has been extended to $(\alpha, \epsilon)$-R\'{e}nyi Pufferfish Privacy (RPP) \cite{RPP2024}, which originally utilized a $W_\infty$ mechanism scaled by the order $\alpha$~\cite[Corollary~3.1]{RPP2024}.
However, because an order $\alpha < \infty$ eases the stringent privacy guarantee of $\epsilon$-pufferfish privacy, and given that $\alpha$ plays a functionally identical role in the $W_\alpha$ metric, it is natural to calibrate the mechanism using a Wasserstein distance of the same order $\alpha$. This alignment ensures that the added noise is effectively reduced compared to standard $\epsilon$-pufferfish privacy.
This intuition motivated the $W_\alpha$ mechanism proposed in \cite[Section 4]{RPP2024}. Nevertheless, their approach necessitates an additional relaxation parameter $\delta$ alongside the Rényi order $\alpha$ \cite[Definition 4.1]{RPP2024}, introducing a double approximation into the privacy guarantee.

In this paper, we propose a $W_\alpha$ mechanism to achieve exact $(\alpha, \epsilon)$-Rényi Pufferfish Privacy (RPP) without requiring further relaxations. Our main contributions are summarized as follows:

\subsubsection{Calibration of Laplace Mechanisms} By applying H\"{o}lder's inequality, we derive a sufficient condition for calibrating Laplace noise via the $\alpha$-Wasserstein metric. Specifically, we show that if the scale parameter $b$ ensures the $W_\alpha$ metric is upper bounded by $\epsilon \frac{\alpha-1}{\alpha}$, $(\alpha,\epsilon)$-R\'{e}nyi pufferfish privacy is satisfied. In the limiting case where $\alpha = \infty$, this recovers the existing $\infty$-Wasserstein mechanism for $\epsilon$-pufferfish privacy \cite{PufferfishWasserstein2017Song}. We further extend this $W_\alpha$ metric approach to the exponential mechanism.

\subsubsection{Gaussian Noise Refinement} For the Gaussian mechanism, we demonstrate that $(\alpha, \epsilon)$-Rényi Pufferfish Privacy is achieved by selecting a variance $\sigma^2$ such that the $W_{\alpha(\alpha-1)}$ metric is upper-bounded by $\frac{\epsilon}{\alpha}$. The appearance of the order $\alpha(\alpha-1)$ stems from the quadratic distance exponent in the Gaussian distribution—the same mathematical foundation behind using the $\ell_2$-norm in the Gaussian mechanism for standard differential privacy. Under deterministic data settings (where Pufferfish privacy reduces to standard differential privacy), this condition aligns exactly with the results in \cite[Corollary 3]{RDP2017} for $(\alpha, \epsilon)$-Rényi differential privacy.

\subsubsection{Utility and Performance Analysis} Experimental results indicate that our proposed $\alpha$-Wasserstein mechanism requires significantly smaller values for $b$ and $\sigma^2$ compared to existing benchmarks in \cite{RPP2024}, leading to a substantial improvement in data utility. Furthermore, we demonstrate that for a fixed $(\alpha, \epsilon)$-R\'{e}nyi pufferfish privacy level, the Gaussian mechanism requires considerably less noise power than the Laplace mechanism when the privacy budget $\epsilon$ is small.

Finally, we outline several directions for future research, including the exploration of operational interpretations for the range $\alpha \in (0,1)$, and the development of $W_2$ mechanisms utilizing Monge's formulation for Gaussian priors.

\subsection{Organization} 
The remainder of this paper is organized as follows. Section~\ref{sec:preliminary} defines the system model and provides the necessary mathematical foundations and privacy definitions. Section \ref{sec:AlphaW} introduces the proposed $\alpha$-Wasserstein mechanism for Laplace, Gaussian, and exponential noise, followed by an evaluation of its performance through experimental results. Finally, Section \ref{sec:future} discusses potential directions for future research and concludes the paper.

\subsection{Notation}

We use capital letters to denote random variables (r.v.s) and lower case letters to denote the elementary event. Calligraphic letters refer to the alphabet of r.v.s. 
For example, $x$ is an instance of r.v. $X$, that takes value in alphabet $\X$. 
We assume $\X \subseteq \Real$, where $\Real$ denotes the set of real numbers. In data privacy settings, $X$ typically represents a numerical query response (e.g., "How many individuals have cancer?").
Denote $P_{X}(x) = \Pr(X = x)$ the probability of outcome $X=x$ when r.v. $X$ takes the value $x$. 
%
We use $P_X = (P_{X}(x) \colon x \in \X)$ to denote a probability distribution and $X \sim P_{X}$ means that r.v. $X$ follows distribution $P_{X}$. 
The support of $P_{X}(\cdot)$  is denoted by $\supp(P_{X}) = \Set{x \in \X \colon  P_X(x) > 0}$.
The expected value of $f(X)$ for some deterministic function $f$ w.r.t. probability $P_{X}$ is denoted by $\E_{X \sim P_{X}}[f(X)] = \int P_{X}(x) f(x) \dif x$.
The conditional probability $P_{Y|X}(y|x)  = \Pr(Y=y|X=x)$ denotes the chances of having $Y=y$ given the outcome $X=x$. $P_{Y|x} = (P_{Y|X}(y|x) \colon y \in \Y)$ refers to the probability distribution of $Y$ conditioned on $X = x$.   
Note that the above notation does not assume a countable alphabet, although this is the case for several examples in this paper where empirical distributions are represented as probability mass functions (e.g., histograms).
Throughout our derivations, we employ general integral operations. When restricted to a countable alphabet, these integrals naturally reduce to summations, and the infimum or supremum operations simplify to minimums and maximums, respectively.

For two probability distributions $P_X$ and $Q_X$, the R\'{e}nyi divergence \cite{Renyi1961_Measures} is
\begin{equation} \label{eq:RD}
	D_{\alpha}(P_X \| Q_X) = \frac{1}{\alpha-1} \log \int \frac{P_X^{\alpha}(x)}{Q_X^{\alpha-1}(x)} \dif x
\end{equation}
where $\alpha \in [0,\infty]$ is referred to as R\'{e}nyi order. In this paper, we assume $P_X \ll Q_X$ so that the Radon–Nikodym derivative is always well defined. 
For extended orders $\alpha= 1$ and $\infty$, we should apply the L'H\^{o}pital's rule to get $D_{1}(P_X \| Q_X) = \E_{X \sim P_X} \big[ \log \frac{P_X(x)}{Q_{X}(x)} \big]$ and $D_{\infty}(P_X \| Q_X) = \log \sup\limits_{x\in\X} \frac{P_{X}(x)}{Q_{X}(x)}$, respectively. Here, $D_{1}$ refers to the Kullback-Leibler divergence.

\section{Preliminary}
\label{sec:preliminary}

We review the Pufferfish privacy framework as originally proposed by Kifer and Machanavajjhala (2012, 2014) \cite{Pufferfish2012KiferConf,Pufferfish2014Kifer}, alongside its extension to the Rényi-divergence-based variant, Rényi Pufferfish Privacy (Pierquin et al., 2024 \cite{RPP2024}). Furthermore, we examine established noise calibration methods based on the recent $W_1$ and $W_\infty$ Wasserstein metric calibration techniques.

\subsection{System Setting and R\'{e}nyi Pufferfish privacy} 
\label{sec:RPP}

Assume that the data to be published, $X$ (e.g., a query response or a column in a table), is statistically correlated with a sensitive secret $S$. Let $P_{X|s, \rho}$ denote the conditional probability distribution of the data $X$ given a secret instance $S = s$, where $\rho$ represents the adversary's prior knowledge—such as the mean and covariance in the case of Gaussian-distributed data. In a multi-adversary environment, different agents may possess distinct prior beliefs $\rho$.
To preserve privacy, we transform $X$ into a randomized output $Y$ before publication. The adversary is assumed to have access only to this sanitized data $Y$, though they may attempt to infer individual secrets by analyzing aggregated statistics from repeated queries.
Let $\SetPair$ define a set of secret pairs $(s_i, s_j)$ specified by the data curator. This set identifies the instances where statistical indistinguishability must be enforced to ensure robust data protection. Any significant discrepancy between the distributions of $Y$ conditioned on $S = s_i$ versus $S = s_j$ could be exploited by an adversary to distinguish between secret states, leading to a privacy breach. This risk motivates a formal privacy definition that imposes an upper bound on the statistical distinguishability between such posterior distributions.

\paragraph{Pufferfish Privacy} For a \emph{privacy budget} $\epsilon>0$, the privatized data $Y$ is said to be $\epsilon$-\emph{pufferfish privacy} if \cite{Pufferfish2014Kifer, Pufferfish2012KiferConf}
\begin{align}\label{eq:PP}
	P_{Y|S}(y|s_i,\rho) \leq e^{\epsilon} P_{Y|S}(y|s_j,\rho), \quad \forall y,  \rho, (s_i,s_j) \in \mathbb{S}.
\end{align} 
Equation \eqref{eq:PP} guarantees an $\epsilon$-level of indistinguishability across all adversarial prior beliefs $\rho$. The formalization of R\'{e}nyi Pufferfish Privacy mirrors the extension of differential privacy to its R\'{e}nyi counterpart, as established in \cite{RDP2017}. 

\paragraph{R\'{e}nyi Pufferfish Privacy}
For a \emph{privacy budget} $\epsilon>0$ and R\'{e}nyi order $\alpha \in [1,\infty]$, the privatized data $Y$ is said to be $(\alpha,\epsilon)$-\emph{R\'{e}nyi pufferfish privacy} in $\SetPair$ if \cite{RPP2024}
\begin{align}\label{eq:RPP}
	D_{\alpha} (P_{Y|s_i,\rho}  \| P_{Y|s_j, \rho}) \leq \epsilon, \quad \forall \rho, (s_i,s_j) \in \mathbb{S}.
\end{align}

For given input distributions, $D_\alpha$ is (strictly) increasing in $\alpha$ \cite[Theorem~3]{Erven2014_JOURNAL}. It reaches maximum at $\alpha = \infty$, where $(\infty,\epsilon)$-R\'{e}nyi pufferfish privacy refers to  
$$D_{\infty} (P_{Y|s_i,\rho} \| P_{Y|s_j, \rho}) =\log  \sup\limits_{y} \frac{P_{Y|S}(y|s_i,\rho) }{P_{Y|S}(y|s_j,\rho)} \leq \epsilon $$
for all $ \rho, (s_i,s_j) \in \mathbb{S}$, equivalent to $\epsilon$-pufferfish privacy. 
This is clear if we rewrite the definition as 
\begin{multline} \label{eq:RDP_GenMean}
	D_{\alpha} (P_{Y|s_i, \rho} \| P_{Y|s_j, \rho} ) \\
	= \log \Big( \E_{Y \sim P_{Y|s_i,\rho}} \Big[ \big( \frac{P_{Y|S}(\cdot|s_i,\rho)}{P_{Y|S}(\cdot|s_j, \rho)} \big)^{\alpha-1}  \Big]     \Big)^{\frac{1}{\alpha-1}} 
\end{multline}
Here, $e^{D_{\alpha} (P_{Y|s_i, \rho} \| P_{Y|s_j, \rho} )}$ is an $(\alpha-1)$-exponent generalized (H\"{o}lder) mean that is monotonically nondecreasing in $\alpha$.

The expression in \eqref{eq:RDP_GenMean} elucidates how R\'{e}nyi Pufferfish Privacy provides a relaxation of the standard Pufferfish framework. At $\alpha = \infty$, the generalized mean locates at the maximum statistical distinguishability, $\frac{P_{Y|S}(y|s_i,\rho)}{P_{Y|S}(y|s_j,\rho)}$, aligning with the core objective of data privacy: protecting against the worst-case, or catastrophic, data breach, irrespective of its frequency.
However, when preventing this worst-case scenario becomes practically infeasible—for instance, when the required noise power severely degrades the utility of the published data—one may trade a degree of privacy for enhanced data utility. By selecting a finite order $\alpha < \infty$, the generalized mean incorporates the statistical distinguishability across the entire support, where the influence of the maximum distinguishability is effectively discounted by its associated probability mass.
Consequently, an upper bound $\epsilon$ on the R\'{e}nyi divergence $D_\alpha$ no longer constrains the instantaneous worst-case ratio, but rather bounds the overall statistical distinguishability in an average sense. Thus, we can satisfy $D_{\alpha} (P_{Y|s_i, \rho} \| P_{Y|s_j, \rho} ) \leq \epsilon$ even if specific events $X = x$ violate the stringent $\epsilon$-Pufferfish constraint, $\frac{P_{Y|S}(y|s_i,\rho)}{P_{Y|S}(y|s_j,\rho)} \leq e^\epsilon$.

The conceptual motivation for relaxing $\alpha$ from $\infty$ parallels the $\delta$-approximation used in $(\epsilon, \delta)$-differential and pufferfish privacy. While $(\epsilon, \delta)$-privacy guarantees that the probability of violating the requirement $\frac{P_{Y|S}(y|s_i,\rho)}{P_{Y|S}(y|s_j,\rho)} \leq e^\epsilon$ is bounded by $\delta$, R\'{e}nyi privacy offers a different, though related, form of relaxation.
Because of this shared goal, $(\alpha, \epsilon)$-pufferfish privacy can always be translated into the $(\epsilon, \delta)$ framework. For instance, according to \cite[Proposition~3]{RDP2017}, a finite order $\alpha$ can be viewed as an increase in the effective privacy budget from $\epsilon$ to $\epsilon + \frac{-\log\delta}{\alpha-1}$ within a $\delta$-approximate setting. Alternatively, an $\alpha < \infty$ can be expressed in terms of the approximation probability itself. By applying the Chernoff bound, for any $\alpha \in (1, \infty)$, 
\begin{equation} \label{eq:Chernoff}
	\Pr \Big( \frac{P_{Y|S}(Y|s_i,\rho)}{P_{Y|S}(Y|s_j,\rho)} > e^{\epsilon} \Big)  \leq e^{(\alpha-1) ( D_{\alpha}(P_{Y|s_i,\rho}, P_{Y|s_j,\rho}) - \epsilon) }.
\end{equation}

Here, $\Pr(\cdot)$ denotes the probability with respect to the distribution $P_{Y|s_i, \rho}$. Recall that $(\epsilon, \delta)$-pufferfish privacy is satisfied if $P_{Y|s_i, \rho}(\mathcal{A}) \leq e^{\epsilon} P_{Y|s_j, \rho}(\mathcal{A}) + \delta$ for all measurable sets $\mathcal{A}$, all priors $\rho$, and all secret pairs $(s_i, s_j) \in \SetPair$ \cite[Section~5]{Ding2022AISTATS}.
Therefore, any $(\alpha,\epsilon)$-R\'{e}nyi pufferfish privacy guarantee such that $D_{\alpha}(P_{Y|s_i,\rho}, P_{Y|s_j,\rho}) \leq \epsilon$ inherently provides the approximation $(\epsilon, e^{(\alpha-1) (D_{\alpha}(P_{Y|s_i,\rho}, P_{Y|s_j,\rho}) - \epsilon)})$-pufferfish privacy. 
It is important to note that these two relaxation methods—selecting a finite $\alpha$ in the Rényi framework or allowing a $\delta$-approximation with $\alpha = \infty$—serve similar purposes. In this paper, we focus on the former, attaining exact $(\alpha, \epsilon)$-R\'{e}nyi pufferfish privacy without introducing an additional $\delta$ parameter.

\subsection{Additive Noise Mechanism} 

A straightforward approach to data sanitization is adding noise to the original data. Let $N$ denote a zero-mean noise variable that is statistically independent of $X$. The randomized output is then generated as $Y = X + N$. When $X$ is an r.v., the resulting probability distribution of $Y$ is determined by the convolution 
\begin{equation} \label{eq:conv}
	P_{Y|S}(y|s,\rho) = \int P_{N} (y-x) P_{X|S}(x|s,\rho) \dif x. 
\end{equation} 
Laplace noise $N \sim \Lap(b)$ follows the probability distribution $P_{N}(z) = \frac{1}{2b}e^{-\frac{|z|}{b}},  \forall z \in \Real$. The scale parameter $b$ indicates the flatness of Laplace distribution and determines noise variance $ 2b^2$. 
For exponential mechanism $N \sim \Exp(\theta)$~\cite[Section~3.3]{CalibNoiseDP}, $c$ is a metric that is nonnegative, symmetric $c(z) = c(-z), \forall z$, and satisfies the triangular inequality $ c(z) \leq c(a) + c(z-a), \forall z,a$.
The noise distribution is $P_{N} (z) \propto e^{ -\eta(\theta) c(z)}$, where $\eta \propto \frac{1}{\theta}$.
By the triangular inequality,
$	P_{N}(y-x) \leq e^{\eta(\theta) c(x-x')}  P_{N}(y-x') , \forall x,x',y $,
where $e^{\eta(\theta) c(x-x')}$ refers to an upper bound on the probability mass transport cost from $x$ to $x'$.
It is clear that Laplace noise is an example of the exponential mechanism when $\eta(\theta) = 1/\theta$ and $c(z) = |z|$.
For Gaussian noise $N \sim \Gauss(\sigma^2)$, the probability distribution is $P_{N}(z) = \frac{1}{\sqrt{2\pi} \sigma}e^{-\frac{z^2}{2\sigma^2}},  \forall z \in \Real$, with the noise variance being $\sigma^2$. 

Noise calibration involves determining the optimal values for the parameters $b$, $\theta$, and $\sigma$ for the Laplace, exponential, and Gaussian mechanisms, respectively. To preserve the utility of the randomized data $Y$, it is essential to minimize the noise power (variance), thereby navigating the privacy-utility tradeoff. Recall that the variance is $2b^2$ for Laplace noise and $\sigma^2$ for Gaussian noise. 
Specifically, the noise parameters must be tuned to the minimum threshold necessary to satisfy the privacy constraint. Excessively large parameters should be avoided, as they unnecessarily deteriorate data utility without providing additional requisite protection.

\subsubsection{Wasserstein Metric} 

For each pair of prior distributions $P_{X|s_i,\rho}$ and $P_{X|s_j,\rho}$, denote $\pi$ a \emph{coupling} joint distribution such that $P_{X|S}(x|s_i,\rho) = \int \pi(x,x') \dif x'$ for all $x$ and $P_{X|S}(x'|s_j,\rho) = \int \pi(x,x') \dif x$ for all $x'$. 
Note that $\pi$ is not unique. 
For $\alpha \in [1,\infty]$ and a nonnegative ground distance (or cost) function $d(\cdot)$, the $\alpha$-Wasserstein distance is 
\begin{equation} \label{eq:W_Alpha}
	W_\alpha (P_{X|s_i,\rho},P_{X|s_j,\rho}) := \Big(  \inf_{\pi}\int d(x-x')^\alpha \dif \pi(x,x^{\prime}) \Big)^{\frac{1}{\alpha}} 
\end{equation}
measuring the minimum cost for transforming the probability mass from $P_{Y|s_i,\rho}$ to $P_{Y|s_j,\rho}$. 
Wasserstein distance $W_\alpha$ is monotonically \emph{increasing} in $\alpha$. 
For $\alpha = \infty$, 
$$W_\infty (P_{X|s_i,\rho},P_{X|s_j,\rho}) =  \inf_{\pi} \sup_{(x,x') \in \supp(\pi)}  d(x-x'). $$

It is important to clarify that the general definition of the Wasserstein metric is not restricted to a specific ground distance $d(\cdot)$. Conventionally, $d$ is non-decreasing (not necessarily linearly) with respect to the difference between $x$ and $x'$, reflecting the underlying cost of transporting probability mass between two points. However, one should not assume that $d(\cdot)$ is limited to the standard $\ell_1$-norm ($d(z) = |z|$), despite its prevalence in data privacy research.
In this paper, we employ generalized ground distance functions to naturally match the concentration properties of specific noise distributions. For example, we use the exponential ground distance $d(z) = e^{\frac{|z|}{b}}$ in Theorem \ref{theorem:AlphaW_Laplace} for Laplace noise. Crucially, because all these variants conform strictly to the general optimal transport formulation in \eqref{eq:W_Alpha}, they are all fundamentally $W_\alpha$ mechanisms—even when the underlying ground distance $d(\cdot)$ deviates from the standard $\ell_1$-norm.
Furthermore, a Wasserstein mechanism is not constrained to standard orders (e.g., $\alpha = 1, 2,$ or $\infty$). In our framework, $\alpha$ in \eqref{eq:W_Alpha} serves a role mathematically analogous to the Rényi order, taking any value in $(1, \infty)$. To prevent confusion with the conventional $\ell_1$-based $W_\alpha$ metric, we explicitly state the underlying ground distance $d(\cdot)$ whenever defining or applying the metric (e.g., noting $d(z) = e^{\frac{\vert{}z\vert{}}{b}}$ for the Laplace metric).

The ground-distance clarification is also essential for understanding our application of the Kantorovich optimal transport plan $\pi^*$. For $\alpha = 1$, the first-order Wasserstein metric reduces to the classic Earth Mover's Distance:
\begin{equation} \label{eq:W_1}
	W_1 (P_{X|s_i,\rho}, P_{X|s_j,\rho})  = \inf_{\pi}\int d(x-x') \dif \pi(x,x'),
\end{equation}
where the minimization is a linear program whose minimizer $\pi^*$ is the Kantorovich optimal transport plan \cite{Villani2009OPT, Santambrogio2015OPT}.
Assuming that the ground distance $d$ is convex, the optimal coupling $\pi^*$ can be evaluated directly from the marginal cumulative distribution functions (CDFs) $F_{X|S}(\cdot|s_i,\rho)$ and $F_{X|S}(\cdot|s_j,\rho)$~\cite{Villani2009OPT, Santambrogio2015OPT}:
\begin{equation} \label{eq:PiStar}
	\pi^*(x,x') = \frac{\dif^2}{\dif x  \dif x'} \min \big\{ F_{X|S}(x|s_i,\rho), F_{X|S}(x'|s_j,\rho) \big\}.
\end{equation}
A concise proof of \eqref{eq:PiStar} based solely on the convexity of $d$ is provided in \cite[Theorem~2.1]{Thorpe2018OPTIntro}. Remarkably, for any non-negative, convex ground distance $d$—such as $d(z) = \max\{1, z\}$ or the exponential ground distance $d(z) = e^{\frac{|z|}{b}}$—the optimal transport plan $\pi^*$ minimizing \eqref{eq:W_1} remains identical, even though the optimal value $W_1$ varies. Consequently, $\pi^*$ in \eqref{eq:PiStar} is determined entirely by the underlying data statistics, independently of the explicit choice of the convex ground distance $d$.
Furthermore, for any $\alpha \in [1,\infty)$, if the composite ground cost $d(\cdot)^\alpha$ is convex, the minimization in \eqref{eq:W_Alpha} reduces to a $W_1$ problem under the effective ground metric $\tilde{d}(\cdot) = d(\cdot)^\alpha$, whose optimal coupling is likewise given by $\pi^*$.
Finally, when the support $\mathcal{X}$ is finite or countable (as is typical in empirical settings), computing $\pi^*$ via \eqref{eq:PiStar} incurs a computational complexity of only $\mathcal{O}(|\X|^2)$ \cite[Appendix~C]{Ding2022AISTATS}.

\section{$\alpha$-Wasserstein Mechanism}
\label{sec:AlphaW}

We maintain consistent notation by using $\alpha$ to denote the order for both the R\'{e}nyi divergence ($D_\alpha$) and the Wasserstein metric ($W_\alpha$), as the parameter serves a functionally analogous role in both frameworks. This notation establishes a direct correspondence between the two measures for any given value of $\alpha$. Given that the $W_\infty$ metric is utilized to calibrate noise for $\epsilon$-pufferfish privacy \cite{PufferfishWasserstein2017Song}, it is natural to anticipate a corresponding $W_\alpha$ mechanism for $(\alpha, \epsilon)$-R\'{e}nyi pufferfish privacy. In this section, we formally validate this intuition by proposing $\alpha$-Wasserstein mechanisms for Laplace and Gaussian noise, as well as an exponential mechanism, for the range $\alpha \in (1, \infty)$.

\subsection{Laplace Noise}
\label{sec:Lap}

The Laplace mechanism was the inaugural method proposed for achieving differential privacy, introduced concurrently with the framework's formal definition in \cite{CalibNoiseDP}. Its prominence stems from the fact that the privacy requirement can be satisfied through straightforward arithmetic properties of the Laplace distribution. Consequently, it remains the most widely adopted additive noise mechanism across various extensions and variations of the differential privacy framework.
In the context of pufferfish privacy, Song et al. \cite{PufferfishWasserstein2017Song} first demonstrated that calibrating the Laplace scale parameter to the $\infty$-Wasserstein distance between discriminative secrets ensures $\epsilon$-pufferfish privacy—a result later extended to a Kantorovich ($W_1$) mechanism in \cite{Ding2022AISTATS}. Intuitively, this suggests that an $\alpha$-Wasserstein mechanism should exist for the R\'{e}nyi Pufferfish Privacy framework. In this section, we derive a method for calibrating the scale parameter using the $W_\alpha$ metric to satisfy $(\alpha, \epsilon)$-R\'{e}nyi pufferfish privacy, and we demonstrate that our approach generalizes the existing $W_\infty$ mechanism.

\begin{theorem} \label{theorem:AlphaW_Laplace}	
	For all $\alpha \in (1,\infty]$, adding Laplace noise $N \sim \Lap(b)$ attains ($\alpha$,$\epsilon$)-R\'{e}nyi pufferfish privacy in $Y$
	if $b$ satisfies
	\begin{equation} \label{eq:theorem:AlphaW_Laplace}
		\int e^{\alpha\frac{|x-x'|}{b}}  \dif \pi^*(x,x') \leq e^{(\alpha-1) \epsilon}
	\end{equation}
	for all $(s_i,s_j) \in \SetPair$ and $\rho$.
\end{theorem}

\begin{IEEEproof}
	For each secret pair $(s_i,s_j) \in \SetPair$ and prior belief $\rho$, using the definition of R\'{e}nyi divergence and convolution \eqref{eq:conv}, for Laplace noise, we have the following inequalities.
	\begin{align}
		D_{\alpha} & (P_{Y|s_i,\rho} \| P_{Y|s_j,\rho}) \nonumber \\
		&= \frac{1}{\alpha-1} \log \int 
		\frac{ \big( \int P_N(y-x) P_{X|S} (x|s_i) \dif x \big)^\alpha }{ \big( \int P_N(y-x') P_{X|S} (x'|s_j) \dif x' \big)^{\alpha-1} } \dif y \nonumber  \\
		&=   \frac{1}{\alpha-1} \log \int \frac{1}{2b}
		\frac{  \big( \int e^{-\frac{|y-x|}{b}} \dif \pi(x,x') \big)^\alpha }{ \big( \int e^{-\frac{|y-x'|}{b}} \dif \pi(x,x')  \big)^{\alpha-1}  }		\dif y \label{eq:AllPi} 
	\end{align}
	for all $\alpha \in (0,\infty)$. Note that equation~\eqref{eq:AllPi} holds for all joint probability $\pi$. 
	The numerator in \eqref{eq:AllPi}  is upper bounded by 
	\begin{align}
		& \Big( \int e^{-\frac{|y-x|}{b}} \dif \pi(x,x') \Big)^\alpha \nonumber \\
		& \leq   \Big( \int e^{-\frac{|y-x'|}{b}} e^{\frac{|x-x'|}{b}} \dif \pi(x,x') \Big)^\alpha  \label{eq:TriIneq} \ \\
		& =   \Big( \int e^{-\frac{|y-x'|}{b} \frac{1}{\alpha}} e^{-\frac{|y-x'|}{b} \frac{\alpha-1}{\alpha}}  e^{\frac{|x-x'|}{b}} \dif \pi(x,x') \Big)^\alpha \nonumber \\
		&\leq \Big( \int e^{-\frac{|y-x'|}{b}} \dif \pi(x,x')  \Big)^{\alpha-1}  \int e^{-\frac{|y-x'|}{b}}  e^{\alpha \frac{|x-x'|}{b}} \dif \pi(x,x')  \label{eq:HolderInq}
	\end{align}
	for all $\alpha \in (1,\infty)$, where~\eqref{eq:TriIneq} is the triangular inequality, and~\eqref{eq:HolderInq} is an application of the H\"{o}lder's inequality. Here, $\alpha > 1$ and $\frac{\alpha}{\alpha-1} > 1$ are H\"{o}lder conjugates such that $\frac{1}{\alpha} + \frac{\alpha-1}{\alpha} = 1$. 
	Apply the upper bound \eqref{eq:HolderInq} to \eqref{eq:AllPi}, we have
	\begin{align}
		&D_{\alpha}  (P_{Y|s_i,\rho} \| P_{Y|s_j,\rho}) \nonumber \\
		&\leq  \frac{1}{\alpha-1} \log \int \frac{1}{2b}  
		\int e^{-\frac{|y-x'|}{b}}  e^{\alpha \frac{|x-x'|}{b}} \dif \pi(x,x') \dif y \nonumber \\
		&= 	\frac{1}{\alpha-1} \log  \int \Big( \int P_{N}(y-x') \dif y \Big)  e^{\alpha \frac{|x-x'|}{b}} \dif \pi(x,x') \\
		&= \frac{1}{\alpha-1} \log  \int e^{\alpha \frac{|x-x'|}{b}} \dif \pi(x,x'). \label{eq:AlphaW_Lap_Raw}
	\end{align}	
	Then, it suffices to request~\eqref{eq:AlphaW_Lap_Raw} upper bounded by $\epsilon$. In order to obtain the smallest scale parameter $b$ that satisfies this condition, we apply a minimization of the integral in~\eqref{eq:AlphaW_Lap_Raw} over all joint probability $\pi$ and formulate the inequality
	\begin{equation} \label{eq:AlphaW_Lap_Min}
		\inf_{\pi}  \int e^{\alpha \frac{|x-x'|}{b}} \dif \pi(x,x') \leq e^{(\alpha-1)\epsilon}.
	\end{equation}
	For each $\alpha \in (1,\infty)$, the LHS of  \eqref{eq:AlphaW_Lap_Min} is a $W_1$ distance for convex $e^{\alpha \frac{|\cdot|}{b}}$, and its minimizer is the Kantorovich optimal plan $\pi^*$~\cite{Villani2009OPT,Santambrogio2015OPT}. Substituting $\pi^*$ to \eqref{eq:AlphaW_Lap_Min} and request this sufficient condition satisfied for all secret pairs $(s_i,s_j)$ and $\rho$, we have Theorem~\ref{theorem:AlphaW_Laplace}.  
\end{IEEEproof}

\subsubsection{Computation of $b$ and Comparison}
\label{sec:Compute_b}

While the sufficient condition in Theorem \ref{theorem:AlphaW_Laplace} follows directly from H\"{o}lder's inequality, it is instructive to detail its practical application for calibrating noise on real-world datasets. 
For a given $\alpha > 1$, the LHS of \eqref{eq:theorem:AlphaW_Laplace}, namely $\int e^{\alpha\frac{|x-x'|}{b}}  \dif \pi^*(x,x')$, is strictly monotonically decreasing in $b$ for any fixed secret pair $(s_i, s_j)$. Thus, the minimal noise scale $b$ satisfying the inequality in \eqref{eq:theorem:AlphaW_Laplace} is obtained at equality. To guarantee $(\alpha, \epsilon)$-Rényi pufferfish privacy across the entire framework, this scale parameter must be maximized over all secret pairs and prior distributions.
Specifically, the calibration procedure proceeds as follows: For each prior distribution $\rho$ and each secret pair $(s_i, s_j) \in \SetPair$, find the unique root $b(s_i, s_j)$ to the equation:
\begin{equation} \label{eq:RootFind}
	\int e^{\alpha\frac{|x-x'|}{b(s_i,s_j)}}  \dif \pi^*(x,x') = e^{(\alpha-1) \epsilon};
\end{equation}
Set the Laplace scale parameter $b$ to the maximum value over all priors $\rho$ and secret pairs by $b = \max_{\rho, \, (s_i,s_j) \in \mathcal{S}_{\text{pair}}} b(s_i,s_j);$
Adding Laplace noise $N \sim \Lap(b)$ to the query output $X$ yields a mechanism output $Y = X + N$ that satisfies $(\alpha, \epsilon)$-R\'{e}nyi Pufferfish Privacy.

Solving \eqref{eq:RootFind} inherently requires a numerical root-finding algorithm. A practical choice is the modified Brent's method proposed in \cite{Yang2026Noise, Ding2025Multi}, which adapts standard bracketed root-finding techniques \cite{Brent1971, Suli2003book} to iteratively refine lower and upper bounds on $b(s_i, s_j)$. For a complete algorithmic specification, we refer the reader to \cite{Yang2026Noise}; however, any stable root-finding routine (e.g., Newton-Raphson or secant methods) can be readily applied.
When the support $\mathcal{X}$ is finite or countable, the integral in \eqref{eq:RootFind} simplifies to a finite summation:
$\sum_{(x,x') \in \supp(\pi^*)} \pi^*(x,x') e^{\alpha\frac{|x-x'|}{b(s_i,s_j)}} = e^{(\alpha-1) \epsilon}, $
which structurally resembles a generalized polynomial in $e^{1/b}$. Because $\supp(\pi^*)$ is typically sparse and significantly smaller than the full Cartesian product $\mathcal{X}^2$, evaluating such a summation during each iteration incurs minimal computational overhead within the root-finding loop.

Similar to Theorem~\ref{theorem:AlphaW_Laplace}, the existing noise calibration method for R\'{e}nyi pufferfish privacy also lacks a closed-form solution. Specifically, it was established in \cite[Corollary 3.1]{RPP2024} that a Laplace scale parameter $b$ guarantees $(\alpha, \epsilon)$-Rényi Pufferfish Privacy if it satisfies:
\begin{multline} \label{eq:ICML_Lap}
	\frac{\alpha}{2\alpha-1} e^{\frac{\alpha-1}{b} W_{\infty}(P_{X|s_i,\rho},P_{X|s_j,\rho}) } \\
	+ \frac{\alpha-1}{2\alpha-1} e^{-\frac{\alpha}{b} W_{\infty}(P_{X|s_i,\rho},P_{X|s_j,\rho}) } \leq e^{(\alpha-1)\epsilon}.
\end{multline}
For a given privacy budget $\epsilon$, the minimal required scale parameter $b$ must likewise be determined via a numerical root-finding algorithm.
The result \eqref{eq:ICML_Lap} was derived in \cite[Corollary 3.1]{RPP2024} by applying the shift reduction lemma \cite[Lemma 20]{Feldman2018FOCS} to obtain the shifted R\'{e}nyi divergence \cite[Definition 8]{Feldman2018FOCS}.
Essentially, this constitutes an $\infty$-Wasserstein mechanism analogous to the R\'{e}nyi differential privacy framework in \cite[Proposition 6]{RDP2017}, with the $\ell_1$-sensitivity replaced by the $W_{\infty}$ distance. This alignment is expected, as the maximum $\ell_1$-norm in the pufferfish setting corresponds exactly to the $\infty$-Wasserstein distance. However, because the $W_{\alpha}$ metric is monotonically non-decreasing with respect to $\alpha$, relying on the $W_{\infty}$ distance inevitably necessitates a larger noise scale to satisfy the privacy constraint.
Experimental results in Figures~\ref{fig:Lap_0_5} and \ref{fig:Gauss_Lap_1}  demonstrate that our proposed $\alpha$-Wasserstein mechanism, as defined in Theorem \ref{theorem:AlphaW_Laplace}, requires a significantly smaller scale parameter $b$ compared to \cite[Corollary 3.1]{RPP2024}.

One approach to improving data utility is to relax the Wasserstein mechanism from $\alpha = \infty$ to a finite $\alpha < \infty$. To this end, \cite[Section 4]{RPP2024} introduced a $\delta$-approximation for R\'{e}nyi Pufferfish Privacy, formally defined by the triplet $(\alpha, \epsilon, \delta)$-R\'{e}nyi Pufferfish Privacy~\cite[Definition 4.1]{RPP2024}. Subsequently, a sufficient condition based on the $\alpha$-Wasserstein metric was proposed in \cite[Theorem 4.3]{RPP2024} for general cases. This was achieved by approximating the shift reduction in the post-processing of Rényi divergence \cite[Lemma 4.1]{RPP2024}.
However, as discussed in Section \ref{sec:RPP}, the R\'{e}nyi measure is itself a relaxation of the stringent $\epsilon$-pufferfish privacy constraint. Specifically, it allows for a breach probability bounded by $e^{(\alpha-1) ( D_{\alpha}(P_{Y|s_i,\rho}, P_{Y|s_j,\rho}) - \epsilon)}$, as shown in \eqref{eq:Chernoff}. Consequently, there is no inherent need to further approximate R\'{e}nyi pufferfish Privacy, as the framework is already an approximation by design. Introducing an additional parameter $\delta$ further eases the privacy constraint, which may lead to unintended consequences.
For instance, an $(\alpha, \epsilon, \delta)$-R\'{e}nyi Pufferfish Privacy guarantee may be equivalent to an $(\epsilon, \delta')$-pufferfish privacy bound where\footnote{We conjecture that the resulting approximation probability is additive in the $(\alpha, \epsilon, \delta)$-R\'{e}nyi pufferfish Privacy framework.} 
$$ \delta' = e^{(\alpha-1) ( D_{\alpha}(P_{Y|s_i,\rho}, P_{Y|s_j,\rho}) - \epsilon)} + \delta.$$
In such cases, $\delta$ must be selected with extreme care; if the combined $\delta'$ approaches or exceeds 1, the privacy guarantee becomes vacuous. It is evident that applying relaxations via both $\alpha$ and $\delta$ complicates the calculation of the cumulative privacy loss. Therefore, Theorem~\ref{theorem:AlphaW_Laplace} and the subsequent results in this work focus exclusively on relaxation through the R\'{e}nyi order $\alpha$.

\subsubsection{Closed-from Solution--Example}
\label{sec:LapRelax}

When a closed-form noise calibration formula is required, the sufficient condition in Theorem \ref{theorem:AlphaW_Laplace} can be relaxed. The core idea is to bound the LHS of \eqref{eq:theorem:AlphaW_Laplace} by a strictly decreasing, invertible function $f_\alpha(b)$:
$$\int e^{\alpha\frac{|x-x'|}{b}}  \dif \pi^*(x,x') \leq f_\alpha(b),$$ 
which allows the scale parameter to be evaluated explicitly in closed form as $b = f_{\alpha}^{-1}(e^{(\alpha-1) \epsilon})$.. We illustrate this relaxation strategy in the following corollary.

\begin{corollary}[Relaxed Sufficient Condition--Laplace] \label{coro:LapRelax}
	For all $\alpha \in (1,\infty)$, adding Laplace noise $N \sim \Lap(b)$ attains ($\alpha$,$\epsilon$)-R\'{e}nyi pufferfish privacy in $Y$
	if
	\begin{equation} \label{eq:coro:AlphaW_Laplace_Relax}
		b = \alpha \cdot \max_{\rho, (s_i,s_j) \in \SetPair} \frac{ \sup_{(x,x') \in \supp(\pi^*)} |x-x'| }{ \log \big( \frac{e^{(\alpha-1)\epsilon}-1}{1-p^*} + 1 \big) } 
	\end{equation}
	where $p^* = \int \pi^{*}(x,x) \dif x$.  
\end{corollary}
\begin{IEEEproof}
	Consider the LHS of inequality \eqref{eq:theorem:AlphaW_Laplace}, while $p^*$ denotes the probability w.r.t. $\pi^*$ that $|x-x'| = 0$, we have probability $1-p^*$ that $0 < |x-x'|  \leq \sup_{(x,x') \in \supp(\pi^*)} |x-x'|$. Therefore, 
	\begin{equation*}
		\int e^{\alpha\frac{|x-x'|}{b}}  \dif \pi^*(x,x') \leq (1-p^*) e^{\alpha \frac{\sup_{(x,x') \in \supp(\pi^*)} |x-x'| }{b}} +p^*
	\end{equation*}
	It suffices to request the RHS to be no greater than $e^{(\alpha-1)\epsilon}$ to attain ($\alpha$,$\epsilon$)-R\'{e}nyi pufferfish privacy, and we have $b \geq \alpha \cdot \frac{ \sup_{(x,x') \in \supp(\pi^*)} |x-x'| }{ \log \big( \frac{e^{(\alpha-1)\epsilon}-1}{1-p^*} + 1 \big) } $ being a relaxed sufficient condition for a single pair of secrets $(s_i,s_j)$. 
	Maximizing this lower bound on $b$ over all secret pairs and prior $\rho$ and then choosing the equality to get the minimum scale parameter, we have \eqref{eq:coro:AlphaW_Laplace_Relax}. Note that the underlying assumption of Corollary~\ref{coro:LapRelax} is $W_{\infty} > 0$ for $d(z) = |z|$. This is usually the case between two different probability distributions. 
\end{IEEEproof}
%

The upper-bounding function $f_\alpha(b)$ is not unique. For example, for $b \geq 1$ (which typically holds under small privacy budgets $\epsilon$), one can apply Jensen's inequality to obtain:
$\int e^{\alpha\frac{|x-x'|}{b}}  \dif \pi^*(x,x') \leq \exp \big( \frac{1}{b} \int \alpha |x-x'| \dif \pi^*(x,x') \big),$
thereby yielding an alternative closed-form expression for the scale parameter $b$. Naturally, relaxing the condition with any upper bound $f_\alpha(b)$ leads to a larger noise scale $b$ than solving the exact equality in \eqref{eq:RootFind}, with the magnitude of this noise penalty depending directly on the tightness of the bound.

\subsubsection{$\alpha$-Wasserstein Mechanism Interpretation and Generalization}

Although the optimal transport plan $\pi^*$ in \eqref{eq:theorem:AlphaW_Laplace} is formulated similarly to the Kantorovich $W_1$ metric, Theorem~\ref{theorem:AlphaW_Laplace} actually establishes a sufficient condition based on the $W_\alpha$ metric. This relationship becomes evident by rewriting \eqref{eq:AlphaW_Lap_Min} as:
\begin{multline} \label{eq:AlphaW_Lap_Min_Alt}
	W_{\alpha}(P_{X|s_i,\rho}, P_{X|s_j,\rho}) \\
	= \Big( \inf_{\pi}  \int e^{\alpha \frac{|x-x'|}{b}} \dif \pi(x,x') \Big)^{\frac{1}{\alpha}} 
	\leq e^{ \frac{\alpha-1}{\alpha} \epsilon }
\end{multline}
where the ground distance is $d(z) = e^{\frac{|z|}{b}}$ for all $z \in \mathbb{R}$. Under this formulation, the scale parameter $b$ in Theorem~\ref{theorem:AlphaW_Laplace} is effectively calibrated by a $W_\alpha$ distance; hence, we refer to this as the $\alpha$-Wasserstein mechanism. This approach integrates seamlessly with the established $W_\infty$ mechanism for $\epsilon$-pufferfish privacy, providing a unified framework for varying privacy requirements.

\begin{remark}[Generalization] \label{rem:Laplace}
	Setting $\alpha =\infty$ to consider the problem of attaining $\epsilon$-pufferfish privacy in $Y$, we have \eqref{eq:AlphaW_Lap_Min_Alt} being $W_{\infty}(P_{X|s_i}, P_{X|s_j}) \leq e^{\epsilon}$. This is equivalent to 
	\begin{equation} \label{eq:WInf_Eq}
		b \geq \frac{1}{\epsilon} \max_{\rho, (s_i,s_j) \in \SetPair} W_{\infty}(P_{X|s_i,\rho}, P_{X|s_j,\rho}) , 
	\end{equation}
	where $W_{\infty}(P_{X|s_i,\rho}, P_{X|s_j,\rho}) = \inf_{\pi} \sup_{(x,x') \in \supp(\pi)} |x-x'|$ is an $\infty$-Wasserstein metic for $d(z) = |z|, \forall z \in \Real$. \eqref{eq:WInf_Eq} is exactly the $W_\infty$ mechanism proposed in \cite{PufferfishWasserstein2017Song}. See Figure~\ref{fig:app2}.
\end{remark}	
\begin{proof}
	\eqref{eq:theorem:AlphaW_Laplace} is equivalent to 
	$\frac{1}{\alpha} \int e^{\alpha\frac{|x-x'|}{b}}  \dif \pi^*(x,x') \leq e^{\frac{\alpha-1}{\alpha} \epsilon}.$
	The LHS is an $\ell_\alpha$-norm. As $\alpha \rightarrow \infty$, 
	$b \geq \frac{1}{\epsilon}  \sup_{(x,x') \in \supp(\pi^*)} |x-x'|.$ 
	It is established in \cite[Appendix~B]{Ding2022AISTATS} that $\inf_{\pi} \sup_{(x,x') \in \supp(\pi)} |x-x'| = \sup_{(x,x') \in \supp(\pi^*)} |x-x'|$, an essential result that converts a computation of $W_\infty$ to determining the Kantorovich optimal transport plan $\pi^*$~\cite[Lemma~1]{Ding2022AISTATS}. 
\end{proof}

\subsubsection{Exponential Mechanism}
The $\alpha$-Wasserstein mechanism for Laplace noise can be easily extended to the exponential mechanism as follows. The proof is in Appendix~\ref{app:coro:AlphaW_Exp}.  
\begin{corollary} \label{coro:AlphaW_Exp}
	For all $\alpha \in (1,\infty]$, adding exponential mechanism $N \sim \Exp(\theta)$ attains ($\epsilon$,$\alpha$)-R\'{e}nyi pufferfish privacy in $Y$ if $\theta$ satisfies 
	\begin{equation} \label{eq:theorem:AlphaW_Exp}
		\int e^{\alpha \eta(\theta) c(x-x')}  \dif \pi^*(x,x') \leq e^{(\alpha-1) \epsilon}
	\end{equation} 
	for all $(s_i,s_j) \in \SetPair$ and $\rho$. \hfill \IEEEQED
\end{corollary}
This can be reformulated as an $\alpha$-Wasserstein mechanism:
\begin{multline*}
	W_{\alpha}(P_{X|s_i,\rho}, P_{X|s_j,\rho}) 	\\
	= \Big( \inf_{\pi}  \int e^{\alpha \eta(\theta) c(x-x')} \dif \pi(x,x') \Big)^{\frac{1}{\alpha}} \leq e^{ \frac{\alpha-1}{\alpha} \epsilon }
\end{multline*}
for $d(z) = e^{\alpha \eta(\theta) c(z)}$ for all $z \in \Real$. 
In the limiting case where $\alpha = \infty$, we obtain $\theta = \eta^{-1} \big( \epsilon / \sup_{(x,x')\in \supp(\pi^*)} c(x-x') \big) $. This result recovers the Kantorovich-exponential mechanism originally proposed in~\cite[Theorem 1]{Ding2022AISTATS}.

\subsection{Gaussian Noise}

Another widely adopted approach is the Gaussian mechanism. Owing to its sub-Gaussian concentration properties and rapidly decaying tail probabilities, it is often preferred over the Laplace mechanism in applications requiring high data utility and accuracy \cite{Dwork2014book, Balle2018AnalyticGM}. In the context of R\'{e}nyi differential Privacy, the Gaussian mechanism yields a closed-form expression for privacy loss \cite[Proposition 7]{RDP2017}, making noise calibration significantly more straightforward than for the Laplace mechanism \cite[Corollary 3]{RDP2017, Mironov2019RDP}.
Below, we propose an $\alpha$-Wasserstein mechanism for calibrating Gaussian noise to satisfy R\'{e}nyi pufferfish Privacy. We further demonstrate that this formulation generalizes the established R\'{e}nyi differential Privacy results found in \cite[Corollary~3]{RDP2017} to correlated data settings.

\begin{theorem} \label{theorem:AlphaW_Gauss}	
	For $\alpha \in (1,\infty)$, adding Gaussian noise $N \sim \Gauss(\sigma^2)$ attains ($\alpha$,$\epsilon$)-R\'{e}nyi pufferfish privacy in $Y$ if $\sigma^2$ satisfies
	\begin{equation} \label{eq:theorem:AlphaW_Gaussian}
		\int e^{\alpha(\alpha-1)\frac{(x-x')^2}{2\sigma^2}}  \dif \pi^*(x,x') \leq e^{(\alpha-1) \epsilon}
	\end{equation}
	for all $(s_i,s_j) \in \SetPair$ and $\rho$. \hfill \IEEEQED
\end{theorem}
The proof is in Appendix~\ref{app:theorem:AlphaW_Gauss} by H\"{o}lder's inequality.

\subsubsection{Computation of $\sigma^2$ and Closed-form Solution}
\label{sec:Compute_sigma}

Calibrating the Gaussian noise variance $\sigma^2$ in Theorem \ref{theorem:AlphaW_Gauss} follows the same process as in Section \ref{sec:Compute_b}:
For each prior $\rho$ and secret pair $(s_i, s_j) \in \SetPair$, use a root-finding algorithm to solve
$\int e^{\alpha(\alpha-1)\frac{(x-x')^2}{2\sigma(s_i,s_j)^2}}  \dif \pi^*(x,x') = e^{(\alpha-1) \epsilon}$
for $\sigma(s_i,s_j)^2$;
Set the overall noise variance to the maximum across all priors and secret pairs
$ \sigma^2 = \max_{\rho, (s_i,s_j) \in \SetPair} \sigma(s_i,s_j)^2;$
Adding Gaussian noise $N \sim \Gauss(\sigma^2)$ to the query output $X$ ensures that $Y = X + N$ satisfies $(\alpha,\epsilon)$-Rényi Pufferfish Privacy.
To obtain an explicit closed-form formula for $\sigma^2$, one can apply the relaxation technique from Section \ref{sec:LapRelax}. The following corollary illustrates this approach, paralleling Corollary~\ref{coro:LapRelax} for Laplace noise (proof omitted as it is essentially identical).

\begin{corollary}[Relaxed Sufficient Condition--Gaussian] \label{coro:GaussRelax}
	For all $\alpha \in (1,\infty)$, adding Gaussian noise $N \sim \Gauss(\sigma^2)$ attains ($\alpha$,$\epsilon$)-R\'{e}nyi pufferfish privacy in $Y$ if
	\begin{equation} \label{eq:coro:AlphaW_Gauss_Relax}
		\sigma^2 = \frac{\alpha(\alpha-1)}{2} \cdot \max_{\rho, (s_i,s_j) \in \SetPair} \frac{ \sup_{(x,x') \in \supp(\pi^*)} (x-x')^2 }{ \log \big( \frac{e^{(\alpha-1)\epsilon}-1}{1-p^*} + 1 \big) } 
	\end{equation}
	where $p^* = \int \pi^{*}(x,x) \dif x$. \hfill \IEEEQED
\end{corollary}

As with the Laplace case (Section \ref{sec:LapRelax}), Corollary \ref{coro:GaussRelax} avoids root-finding at the cost of a larger noise variance $\sigma^2$ than Theorem \ref{theorem:AlphaW_Gauss}. This variance inflation depends directly on the optimal transport plan $\pi^*$ (i.e., the underlying distributions $X|s_i$ and $X|s_j$), and can be tightened using alternative relaxation methods.

\subsubsection{$\alpha(\alpha-1)$-Wasserstein Mechanism Interpretation and Generalization}

Theorem~\ref{theorem:AlphaW_Gauss} is in fact a $W_{\alpha(\alpha-1)}$ mechanism. This is clear if we rewrite~\eqref{eq:theorem:AlphaW_Gaussian} to 
\begin{multline} \label{eq:AlphaW_Gaussian_Min_Alt}
	W_{\alpha(\alpha-1)}(P_{X|s_i}, P_{X|s_j}) \\
	= \Big( \inf_{\pi}  \int e^{\alpha (\alpha-1) \frac{(x-x')^2}{2 \sigma^2}} \dif \pi(x,x') \Big)^{\frac{1}{\alpha(\alpha-1)}}  \leq e^{ \frac{ \epsilon }{\alpha} }
\end{multline}
where $d(z) = e^{\frac{z^2}{2\sigma^2}}$ for all $z \in \Real$. 
The elevated order $\alpha(\alpha-1)$ reflects the squared distance exponent $(x-x')^2$ inherent to the Gaussian distribution. Because we calibrate the noise variance $\sigma^2$ directly—rather than the scale parameter $b$ used for Laplace noise (where variance is $2b^2$)—the Wasserstein metric order naturally scales quadratically.

\begin{remark}[Generalizing from R\'{e}nyi Differential Privacy]
	When the adversary's prior knowledge indicates that the data is deterministic—meaning $P_{X|s_i, \rho}$ and $P_{X|s_j, \rho}$ are point masses centered at distinct values $\mu_i$ and $\mu_j$, respectively—R\'{e}nyi Pufferfish Privacy reduces to standard R\'{e}nyi Differential Privacy. In this scenario, the condition in \eqref{eq:theorem:AlphaW_Gaussian} simplifies to:
	\begin{equation}  \label{eq:RDP_Gauss}
		\alpha \frac{(\mu_i- \mu_j)^2}{2\sigma^2} \leq \epsilon. 
	\end{equation}
	The LHS of this equation represents the R\'{e}nyi divergence between two Gaussian distributions sharing a common variance $\sigma^2$ \cite[Proposition 7]{RDP2017}.	
	By defining the $\ell_1$-sensitivity as $\triangle = \max_{\rho, (s_i,s_j) \in \SetPair} |\mu_i - \mu_j| $, we obtain the closed-form solution $\sigma^2 = \alpha \frac{\triangle^2}{2\epsilon}$. This result is identical to the Gaussian noise calibration method proposed in~\cite[Corollary 3]{RDP2017} for achieving $(\alpha, \epsilon)$-R\'{e}nyi differential privacy.
\end{remark}

Consistent with the framework in \cite[Corollary 3]{RDP2017}, our approach does not require additional relaxations—such as the $\delta$-approximation introduced in \cite[Definition 4.1]{RPP2024}—to calibrate Gaussian noise for R\'{e}nyi pufferfish privacy. 
Replacing the mean difference by $\infty$-Wasserstein metric for $d(z) = |z|$ for all $z \in \Real$, \cite[Corollary 3.1]{RPP2024} proposes a sufficient condition 
\begin{equation} \label{eq:ICML_Gauss}
	\sigma^2 \geq \frac{\alpha W_{\infty}(P_{X|s_i,\rho}, P_{X|s_j,\rho})^2 }{2\epsilon}
\end{equation}
for attaining ($\alpha$,$\epsilon$)-R\'{e}nyi pufferfish privacy. 
Experimental results presented in Figure~\ref{fig:Gauss_Lap} demonstrate that our proposed $\alpha(\alpha-1)$-Wasserstein mechanism, as defined in Theorem~\ref{theorem:AlphaW_Gauss}, requires a significantly smaller variance $\sigma^2$ compared to \cite[Corollary 3.1]{RPP2024}.

\begin{figure*}[t]
	\centering
	\subfloat[\texttt{Adult} dataset for secret pair set in \eqref{eq:SetPair_adult}; the last plot is the maximum over the first three plots. \label{fig:Lapace_adult}]{
		\scalebox{0.57}{\input{figures/real-world/Laplace/adult/epsilon_0_5/pair_1_three_way.tex}}
		\scalebox{0.57}{\input{figures/real-world/Laplace/adult/epsilon_0_5/pair_2_three_way.tex}}
		\scalebox{0.57}{\input{figures/real-world/Laplace/adult/epsilon_0_5/pair_3_three_way.tex}}
		\scalebox{0.57}{\input{figures/real-world/Laplace/adult/epsilon_0_5/maximum_across_three_pairs.tex}}}\\
	\subfloat[\texttt{Heart Disease} dataset for secret pair set in \eqref{eq:SetPair_heart}; the last plot is the maximum over the first three plots. \label{fig:Laplace_heart}]{
		\scalebox{0.57}{\input{figures/real-world/Laplace/heart/epsilon_0_5/pair_1_three_way.tex}} 
		\scalebox{0.57}{\input{figures/real-world/Laplace/heart/epsilon_0_5/pair_2_three_way.tex}}  
		\scalebox{0.57}{\input{figures/real-world/Laplace/heart/epsilon_0_5/pair_3_three_way.tex}} 
		\scalebox{0.57}{\input{figures/real-world/Laplace/heart/epsilon_0_5/maximum_across_three_pairs.tex}}} \\
	\subfloat[\texttt{Student Performance} dataset for secret pair set in \eqref{eq:SetPair_student}; the last plot is the maximum over the first three plots; Theorem~\ref{theorem:AlphaW_Laplace} overlaps Corollary~\ref{coro:LapRelax} in 'goout (2 vs 3)' (2nd plot). \label{fig:Laplace_student}]{
		\scalebox{0.57}{\input{figures/real-world/Laplace/student/epsilon_0_5/pair_1_three_way.tex}} 
		\scalebox{0.57}{\input{figures/real-world/Laplace/student/epsilon_0_5/pair_2_three_way.tex}}
		\scalebox{0.57}{\input{figures/real-world/Laplace/student/epsilon_0_5/pair_3_three_way.tex}}
		\scalebox{0.57}{\input{figures/real-world/Laplace/student/epsilon_0_5/maximum_across_three_pairs.tex}}}
	\caption{Laplace noise calibration: experimental results on three datasets in the UCI machine learning repository~\cite{UCI2007} for $\epsilon = 0.5$, comparing the proposed $W_\alpha$-mechanism in Theorem~\ref{theorem:AlphaW_Laplace} and the closed-form solution in Corollary~\ref{coro:LapRelax} to the existing method \eqref{eq:ICML_Lap} \cite[Corollary~3.1]{RPP2024}.  } \label{fig:Lap_0_5}
\end{figure*}

\begin{figure*}[t]
	\subfloat[\texttt{Adult} dataset for secret pair set in \eqref{eq:SetPair_adult}; the last plot is the maximum over the first three plots. \label{fig:Gauss_adult}]{
		\scalebox{0.57}{\input{figures/real-world/Gaussian/adult/epsilon_0_5/pair_1_three_way.tex}}
		\scalebox{0.57}{\input{figures/real-world/Gaussian/adult/epsilon_0_5/pair_2_three_way.tex}}
		\scalebox{0.57}{\input{figures/real-world/Gaussian/adult/epsilon_0_5/pair_3_three_way.tex}}
		\scalebox{0.57}{\input{figures/real-world/Gaussian/adult/epsilon_0_5/maximum_across_three_pairs.tex}}}\\
	\subfloat[\texttt{Heart Disease} dataset for secret pair set in \eqref{eq:SetPair_heart}; the last plot is the maximum over the first three plots. \label{fig:Gauss_heart}]{
		\scalebox{0.57}{\input{figures/real-world/Gaussian/heart/epsilon_0_5/pair_1_three_way.tex}}
		\scalebox{0.57}{\input{figures/real-world/Gaussian/heart/epsilon_0_5/pair_2_three_way.tex}} 
		\scalebox{0.57}{\input{figures/real-world/Gaussian/heart/epsilon_0_5/pair_3_three_way.tex}}
		\scalebox{0.57}{\input{figures/real-world/Gaussian/heart/epsilon_0_5/maximum_across_three_pairs.tex}}} \\
	\subfloat[\texttt{Student Performance} dataset for secret pair set in \eqref{eq:SetPair_student}; the last plot is the maximum over the first three plots; Theorem~\ref{theorem:AlphaW_Gauss} overlaps Corollary~\ref{coro:GaussRelax} in 'goout (2 vs 3)' (2nd plot). \label{fig:Gauss_student}]{
		\scalebox{0.57}{\input{figures/real-world/Gaussian/student/epsilon_0_5/pair_1_three_way.tex}}
		\scalebox{0.57}{\input{figures/real-world/Gaussian/student/epsilon_0_5/pair_2_three_way.tex}} 
		\scalebox{0.57}{\input{figures/real-world/Gaussian/student/epsilon_0_5/pair_3_three_way.tex}} 
		\scalebox{0.57}{\input{figures/real-world/Gaussian/student/epsilon_0_5/maximum_across_three_pairs.tex}}}
	\caption{Gaussian noise calibration: experimental results on three datasets in the UCI machine learning repository~\cite{UCI2007} for $\epsilon = 0.5$, comparing the proposed $W_\alpha$-mechanism in Theorem~\ref{theorem:AlphaW_Laplace} and the closed-form solution in Corollary~\ref{coro:LapRelax} to the existing method \eqref{eq:ICML_Gauss} \cite[Corollary~3.1]{RPP2024}.  } \label{fig:Gauss_0_5}
\end{figure*}

\begin{figure*}[t]
	\subfloat[$\alpha=2$: the noise variance is the maximum over three secret pairs in \eqref{eq:SetPair_adult} for \texttt{Adult}, \eqref{eq:SetPair_heart} for \texttt{Heart Disease} and \eqref{eq:SetPair_student} for \texttt{Student Performance} \label{fig:Comp_alpha2}]{
		\scalebox{0.6}{\input{figures/Gauss_Lap/adult/alpha_2_epsilon/maximum_across_three_pairs.tex}}
		\scalebox{0.6}{\input{figures/Gauss_Lap/heart/alpha_2_epsilon/maximum_across_three_pairs.tex}}
		\scalebox{0.6}{\input{figures/Gauss_Lap/student/alpha_2_epsilon/maximum_across_three_pairs.tex}}
	}		\\
	\subfloat[$\alpha=4$: the noise variance is the maximum over three secret pairs in \eqref{eq:SetPair_adult} for \texttt{Adult}, \eqref{eq:SetPair_heart} for \texttt{Heart Disease} and \eqref{eq:SetPair_student} for \texttt{Student Performance}  \label{fig:Comp_alpha4}]{
		\scalebox{0.6}{\input{figures/Gauss_Lap/adult/alpha_4_epsilon/maximum_across_three_pairs.tex}}
		\scalebox{0.6}{\input{figures/Gauss_Lap/heart/alpha_4_epsilon/maximum_across_three_pairs.tex}}
		\scalebox{0.6}{\input{figures/Gauss_Lap/student/alpha_4_epsilon/maximum_across_three_pairs.tex}}
	}		
	\caption{Comparison of noise variances between the proposed Laplace mechanism ($2b^2$) and Gaussian mechanism ($\sigma^2$). Here, $b$ is computed via Theorem~\ref{theorem:AlphaW_Laplace} and Corollary~\ref{coro:LapRelax}, while $\sigma^2$ is determined via Theorem~\ref{theorem:AlphaW_Gauss} and Corollary~\ref{coro:GaussRelax}.} \label{fig:Gauss_Lap}
\end{figure*}

\begin{figure*}[t]
	\scalebox{0.57}{\input{figures/real-world/Laplace/adult/epsilon_1/pair_1_three_way.tex}}
	\scalebox{0.57}{\input{figures/real-world/Laplace/adult/epsilon_1/pair_2_three_way.tex}}
	\scalebox{0.57}{\input{figures/real-world/Laplace/adult/epsilon_1/pair_3_three_way.tex}}
	\scalebox{0.57}{\input{figures/real-world/Laplace/adult/epsilon_1/maximum_across_three_pairs.tex}} 
	\scalebox{0.57}{\input{figures/real-world/Laplace/heart/epsilon_1/pair_1_three_way.tex}} 
	\scalebox{0.57}{\input{figures/real-world/Laplace/heart/epsilon_1/pair_2_three_way.tex}} 
	\scalebox{0.57}{\input{figures/real-world/Laplace/heart/epsilon_1/pair_3_three_way.tex}} 
	\scalebox{0.57}{\input{figures/real-world/Laplace/heart/epsilon_1/maximum_across_three_pairs.tex}}
	\scalebox{0.57}{\input{figures/real-world/Laplace/student/epsilon_1/pair_1_three_way.tex}} 
	\scalebox{0.57}{\input{figures/real-world/Laplace/student/epsilon_1/pair_2_three_way.tex}} 
	\scalebox{0.57}{\input{figures/real-world/Laplace/student/epsilon_1/pair_3_three_way.tex}} 
	\scalebox{0.57}{\input{figures/real-world/Laplace/student/epsilon_1/maximum_across_three_pairs.tex}}
	\scalebox{0.57}{\input{figures/real-world/Gaussian/adult/epsilon_1/pair_1_three_way.tex}}
	\scalebox{0.57}{\input{figures/real-world/Gaussian/adult/epsilon_1/pair_2_three_way.tex}}
	\scalebox{0.57}{\input{figures/real-world/Gaussian/adult/epsilon_1/pair_3_three_way.tex}}
	\scalebox{0.57}{\input{figures/real-world/Gaussian/adult/epsilon_1/maximum_across_three_pairs.tex}} 
	\scalebox{0.57}{\input{figures/real-world/Gaussian/heart/epsilon_1/pair_1_three_way.tex}} 
	\scalebox{0.57}{\input{figures/real-world/Gaussian/heart/epsilon_1/pair_2_three_way.tex}} 
	\scalebox{0.57}{\input{figures/real-world/Gaussian/heart/epsilon_1/pair_3_three_way.tex}} 
	\scalebox{0.57}{\input{figures/real-world/Gaussian/heart/epsilon_1/maximum_across_three_pairs.tex}}
	\scalebox{0.57}{\input{figures/real-world/Gaussian/student/epsilon_1/pair_1_three_way.tex}} \ \ 
	\scalebox{0.57}{\input{figures/real-world/Gaussian/student/epsilon_1/pair_2_three_way.tex}} \ \ 
	\scalebox{0.57}{\input{figures/real-world/Gaussian/student/epsilon_1/pair_3_three_way.tex}} \ \ 
	\scalebox{0.57}{\input{figures/real-world/Gaussian/student/epsilon_1/maximum_across_three_pairs.tex}}
	\caption{Experimental results on noise calibration for $\epsilon = 1$; horizontal axis is R\'{e}nyi order $\alpha$.  } \label{fig:Gauss_Lap_1}
\end{figure*}

\subsection{Experiment}
\label{sec:Exp}

Experiments are conducted on three standard benchmark datasets from the UCI Machine Learning Repository \cite{UCI2007}: \texttt{Adult}, \texttt{Heart Disease}, and \texttt{Student Performance}. For each dataset, we evaluate privacy protection across three representative secret pairs.

Extracted from the 1994 U.S. Census Bureau, the \texttt{Adult} dataset contains 48,842 instances across 14 attributes. We set $X$ to be the capital-loss attribute and specify the secret pairs as:
\begin{align} \label{eq:SetPair_adult}
	\SetPair = \big\{ & (\text{'sex is Female'}, \text{'sex is Male'}), \nonumber \\
	&(\text{'race is Amer-Indian-Eskimo'}, \text{'race is White'}), \nonumber \\
	& (\text{'marital-status is Married-AF-spouse'}, \nonumber \\
	&\qquad\qquad\qquad \text{'marital-status is Separated'})\big\}. 
\end{align}
For instance, the pair $(\text{'sex: Female'}, \text{'sex: Male'})$ ensures that releasing capital-loss preserves statistical indistinguishability regarding the participant's gender.
For the \texttt{Heart Disease}, we use the Cleveland subset (303 instances, 14 attributes) collected for predicting heart disease severity. Setting $X$ to the restecg attribute, we define:
\begin{align} \label{eq:SetPair_heart}
	\SetPair = \big\{ & (\text{'sex is 0.0'}, \text{'sex is 1.0'}), (\text{'num is 0'}, \text{'num is 2'}), \nonumber \\
	& (\text{'slope is 2.0'}, \text{'slope is 3.0'})\big\}. 
\end{align}
For the \texttt{Student Performance}, we use the Mathematics subset (395 instances, 33 attributes) evaluating secondary school achievement. We set $X$ to the Dalc (workday alcohol consumption) attribute and define:
\begin{align} \label{eq:SetPair_student}
	\SetPair = \big\{ & (\text{'romantic is no'}, \text{'romantic is yes'}), \nonumber \\
	& (\text{'goout is 2'}, \text{'goout is 3'}), \nonumber \\
	& (\text{'Mjob is services'}, \text{'Mjob is teacher'})\big\}.
\end{align}

\subsubsection{Noise Calibration}

We first evaluate the sufficient condition in Theorem \ref{theorem:AlphaW_Laplace} and the closed-form relaxation in Corollary \ref{coro:LapRelax} for the Laplace mechanism, comparing them against the benchmark calibration method \eqref{eq:ICML_Lap} from \cite[Corollary~3.1]{RPP2024}.
Figure \ref{fig:Lapace_adult} shows the results on the \texttt{Adult} dataset with fixed privacy budget $\epsilon = 0.5$. For Theorem \ref{theorem:AlphaW_Laplace}, we follow Section \ref{sec:Compute_b} by using a root-finding algorithm to compute $b$ for each secret pair in \eqref{eq:SetPair_adult} (the first three subplots) and taking the maximum across all pairs (the final subplot). For Corollary \ref{coro:LapRelax}, $b$ is computed directly via \eqref{eq:coro:AlphaW_Laplace_Relax}, both for individual pairs and after taking the maximum. Parallel results for the \texttt{Heart Disease} and \texttt{Student Performance} datasets are shown in Figures \ref{fig:Laplace_heart} and \ref{fig:Laplace_student}, respectively.

As noted in Section \ref{sec:LapRelax}, Corollary \ref{coro:LapRelax} relaxes the condition in Theorem \ref{theorem:AlphaW_Laplace}, yielding a slightly larger scale parameter $b$. Crucially, both of our methods yield substantially smaller noise scales than the existing approach in \cite[Corollary~3.1]{RPP2024}. 
The minor irregularities near $\alpha = 1$ in Figure \ref{fig:Lap_0_5} stem from the R\'{e}nyi divergence in \eqref{eq:RD} being undefined at $\alpha = 1$. This limit reduces to the Kullback–Leibler divergence ($D_1$), which measures average-case rather than worst-case privacy guarantees and is generally avoided in privacy-preserving mechanisms.
Similarly, Figure \ref{fig:Gauss_0_5} shows that for the Gaussian mechanism, Theorem \ref{theorem:AlphaW_Gauss} and Corollary \ref{coro:GaussRelax} achieve $(\alpha, \epsilon)$-Rényi Pufferfish Privacy with significantly lower noise variance than the benchmark \eqref{eq:ICML_Gauss} from \cite[Corollary~3.1]{RPP2024}. Consistent observations hold under $\epsilon = 1$, as shown in Figure \ref{fig:Gauss_Lap_1}.

\subsubsection{Gaussian vs Laplace}

Figure \ref{fig:Gauss_Lap} evaluates the noise power requirements by comparing the variance of the Laplace mechanism (Theorem~\ref{theorem:AlphaW_Laplace} and Corollary \ref{coro:LapRelax}) against that of the Gaussian mechanism (Theorem \ref{theorem:AlphaW_Gauss} and Corollary \ref{coro:GaussRelax}) for $\alpha = 2$ (Figure \ref{fig:Comp_alpha2}) and $\alpha = 4$ (Figure \ref{fig:Comp_alpha4}).
The results demonstrate that the Gaussian mechanism requires considerably less noise power than the Laplace mechanism when enforcing R\'{e}nyi pufferfish privacy under a small order $\alpha = 2$ (Figure \ref{fig:Comp_alpha2}); this advantage is particularly pronounced in the high-privacy regime where the budget $\epsilon$ is small.
However, as $\alpha$ increases to $4$, R\'{e}nyi pufferfish privacy approaches pure $\epsilon$-Pufferfish Privacy \eqref{eq:PP}. Because the Gaussian mechanism cannot achieve exact $\epsilon$-Pufferfish Privacy at $\alpha = \infty$—offering only an approximate $(\epsilon, \delta)$-guarantee \cite[Theorem~3]{Ding2022AISTATS}—its performance degrades relative to the Laplace mechanism for larger $\alpha$.
Consistent with Figures \ref{fig:Lap_0_5}, \ref{fig:Gauss_0_5}, and \ref{fig:Gauss_Lap_1}, Figure \ref{fig:Gauss_Lap} confirms that the closed-form relaxations in Corollaries~\ref{coro:LapRelax} and \ref{coro:GaussRelax} serve as effective upper bounds on the exact noise requirements of the $W_\alpha$ mechanisms in Theorems~\ref{theorem:AlphaW_Laplace} and \ref{theorem:AlphaW_Gauss}.

\section{Conclusion}
\label{sec:future}

We investigated the calibration of Wasserstein mechanisms to achieve $(\alpha, \epsilon)$-R\'{e}nyi pufferfish privacy. We proposed an $\alpha$-Wasserstein mechanism where the parameters for Laplace and Gaussian noise are calibrated using an upper-bounded $W_\alpha$ metric of the same order $\alpha$.
Experimental results demonstrate that our $\alpha$-Wasserstein mechanism significantly reduces noise compared to existing $W_\infty$-based approaches. The results further verify that the Gaussian mechanism offers superior data utility over the Laplace mechanism when utilizing the Rényi divergence as a privacy relaxation.

The primary results of this paper leverage H\"{o}lder’s inequality for the conjugate exponents $\alpha$ and $\frac{\alpha}{\alpha-1}$. This established technique is a staple of information theory, used in generalized error bounds \cite{Esposito2020_CONF, Esposito2021_JOURNAL} , entropy power inequalities \cite{Rioul2011, Rioul2020ISITA},  and foundational bounds on guessing entropy \cite{Massey1994Guess, Arikan1996Guess}.. Furthermore, Rényi measures of order $\frac{\alpha}{\alpha-1}$ have recently gained prominence in information-theoretic privacy \cite{Liao2019_AlphaLeak, Ding2025ITW}. Beyond its core application to differential and pufferfish privacy, we highlight several promising extensions for future work.

\subsubsection{Multivariate Data Release} When $X = (X_1, X_2, \dots)$ is a vector-valued query output, it is standard to assume that each component $X_i$ responds to an independent (different) query. Consequently, the noise added by Laplace or Gaussian mechanisms has a diagonal covariance matrix \cite[Lemma~8.4, Theorem~A.1]{Dwork2014book}, allowing each dimension to be calibrated independently using the scalar results established in this paper.

\subsubsection{The Regime $\alpha \in (0,1)$}
For $\alpha \in (0,1)$, a sufficient condition for achieving $(\alpha,\epsilon)$-Rényi Pufferfish Privacy under the Laplace mechanism can be similarly derived as:
\begin{equation} \label{eq:Range01}
	\int e^{-\alpha\frac{|x-x'|}{b}}  \dif \pi(x,x') \geq  e^{(\alpha-1) \epsilon}, \quad  \forall \rho, (s_i,s_j) \in \SetPair
\end{equation}
(see Proposition~\ref{prop:Cond_AlpahZeroOne} in Appendix~\ref{app:Cond_AlpahZeroOne}).
However, the operational interpretation of R\'{e}nyi privacy in the range $\alpha \in (0,1)$ remains to be fully established. Indeed, R\'{e}nyi divergence for $\alpha \in (0,1)$ is underexplored in data privacy; this region was explicitly excluded in the foundational formulation of R\'{e}nyi differential privacy \cite[Proposition~10]{RDP2017} and subsequent literature to preserve key properties like H\"{o}lder's inequality.
Nevertheless, by the generalized Pinsker's inequality \cite[Eq.~(1)]{Sason2015}, \cite[Corollary~6]{Gilardoni2010}, for any $\alpha \in (0,1)$ we have:
$	D_{\alpha}(P \| Q) \geq \frac{\alpha}{2} |P-Q|^2, $
where $\frac{1}{2}|P-Q| = \sup_{\mathcal{A}} P(\mathcal{A}) - Q(\mathcal{A})$ denotes the total variation (TV) distance between $P$ and $Q$. Recognizing that TV distance corresponds to the Hockey-Stick divergence 
$E_{e^\epsilon} =\sup_{\mathcal{A} \subseteq \Real} P(\mathcal{A}) - e^{\epsilon} Q(\mathcal{A})$
at $\epsilon = 0$ \cite{Sason2016, Asoodeh2021}, $D_{\alpha}(P \| Q)$ for $\alpha \in (0,1)$ inherently bounds the approximate privacy parameter $\delta$ in $(\epsilon, \delta)$-differential or pufferfish privacy. Fully characterizing this operational link presents a promising direction for future research, where \eqref{eq:Range01} can directly serve as the noise calibration foundation.

\subsubsection{$W_2$ Mechanism for Gaussian Noise and Gaussian Priors}
%
Consider the sufficient condition in Theorem~\ref{theorem:AlphaW_Gauss} for Gaussian noise. 
Because the distance $|x-x'|$ appears quadratically in the exponent, this formulation suggests a natural connection to $W_2$-Wasserstein mechanisms. In particular, when the prior distributions $P_{X|s_i,\rho}$ and $P_{X|s_j,\rho}$ are Gaussian, the $W_2$ metric admits a closed-form expression via Monge's formulation \cite{Dowson1982Frechet, Givens1984NormClass, Takatsu2011Wnorm}. Exploring whether this closed-form structure can yield sharper noise calibration bounds presents an interesting direction for future work.

\begin{figure*}[ht]
	\centering
	\scalebox{0.6}{\input{figures/App/final_equation14_vs_alpha_adult_eps_0p5_multi_b.tex}} 
	\scalebox{0.6}{\input{figures/App/final_equation14_vs_alpha_heart_eps_0p5_multi_b.tex}}
	\scalebox{0.6}{\input{figures/App/final_equation14_vs_alpha_student_eps_0p5_multi_b.tex}}
	\caption{The comparison of 	$W_{\alpha}(P_{X|s_i,\rho}, P_{X|s_j,\rho}) $ with it's upper bound $e^{ \frac{\alpha-1}{\alpha} \epsilon }$ in  \eqref{eq:AlphaW_Lap_Min_Alt} for $\epsilon=0.5$. } \label{fig:app1}
\end{figure*}

\begin{figure*}[ht]
	\centering
	\scalebox{0.6}{\input{figures/App/remark1_convergence_adult_eps_0p5_alpha_ge_10.tex}} \ \ 
	\scalebox{0.6}{\input{figures/App/remark1_convergence_heart_eps_0p5_alpha_ge_10.tex}} \ \ 
	\scalebox{0.6}{\input{figures/App/remark1_convergence_student_eps_0p5_alpha_ge_10.tex}} 
	\caption{The variation of scale parameter $b$ determined by Theorem~\ref{theorem:AlphaW_Laplace} and \cite[Corollary~3.1]{RPP2024} for $\epsilon = 0.5$. They both approach $W_\infty/\epsilon$ mechanism proposed in \cite{PufferfishWasserstein2017Song} as $\alpha \rightarrow \infty$ for attaining $\epsilon$-pufferfish privacy. } \label{fig:app2}
\end{figure*}

\appendices

\section{Proof of Corollary~\ref{coro:AlphaW_Exp}} \label{app:coro:AlphaW_Exp}

\begin{IEEEproof}
	The proof is similar to Theorem~\ref{theorem:AlphaW_Laplace}. We still apply the H\"{o}lder's inequality, but use the triangular inequality, $P_{N_{\theta}}(y-x) \leq e^{\eta(\theta) c(x-x')}  P_{N_{\theta}}(y-x') , \forall x,x',y$. 	
	For $D_{\alpha}  (P_{Y|s_i,\rho} \|  P_{Y|s_j,\rho}) $ in \eqref{eq:AllPi}, we have 
	\begin{align*}
		&\Big( \int P_N(y-x) \dif \pi(x,x') \Big)^\alpha  \\
		&\leq \Big( \int P_N(y-x') e^{\eta(\theta) c(x-x')} \dif \pi(x,x') \Big)^\alpha \\
		&= \Big( \int P_N(y-x')^{\frac{1}{\alpha}} P_N(y-x')^{\frac{\alpha-1}{\alpha}}  e^{\eta(\theta) c(x-x')} \dif \pi(x,x') \Big)^\alpha \\
		&\leq  \Big( \int P_N(y-x') \dif \pi(x,x') \Big)^{\alpha-1}  \cdot \\
		& \qquad\qquad\qquad\qquad \int P_N(y-x') e^{\alpha \eta(\theta) c(x-x')} \dif \pi(x,x').
	\end{align*}
	Then, 
	\begin{align}
		&D_{\alpha} (P_{Y|s_i,\rho} \|  P_{Y|s_j,\rho}) \nonumber \\
		&\leq  \frac{1}{\alpha-1} \log \iint  P_N(y-x') e^{\alpha \eta(\theta) c(x-x')} \dif \pi(x,x') \dif y \nonumber \\
		&=  \frac{1}{\alpha-1} \log \int  \Big( \int  P_N(y-x') \dif y \Big) e^{\alpha \eta(\theta) c(x-x')} \dif \pi(x,x') \nonumber \\
		&=  \frac{1}{\alpha-1} \log \int e^{\alpha \eta(\theta) c(x-x')} \dif \pi(x,x'). \label{eq:AlphaW_Exp_Raw}
	\end{align}
	Upper bounding \eqref{eq:AlphaW_Exp_Raw} by $\epsilon$, and substituting the Kantorovich optimal transport plan $\pi^*$ to minimize the value of $\theta$ over all possible $\pi$, if this sufficient condition is satisfied for all $\rho$ and $(s_i,s_j) \in \SetPair$, Corollary~\ref{coro:AlphaW_Exp} is proved. 
\end{IEEEproof}

\section{Proof of Theorem~\ref{theorem:AlphaW_Gauss}} \label{app:theorem:AlphaW_Gauss}

\begin{IEEEproof}
	For Gaussian noise $N \sim \Gauss(\sigma^2)$,
	\begin{align*}
		D_{\alpha} & (P_{Y|s_i,\rho} \|  P_{Y|s_j,\rho})  \\ 
		&=   \frac{1}{\alpha-1} \log \int \frac{1}{\sqrt{2\pi} \sigma}
		\frac{  \big( \int e^{-\frac{(y-x)^2}{2\sigma^2}} \dif \pi(x,x') \big)^\alpha }{ \big( \int e^{-\frac{(y-x')^2}{2\sigma^2}} \dif \pi(x,x')  \big)^{\alpha-1}  }\dif y  
	\end{align*}
	for all $\alpha \in (1,\infty)$, where
	\begin{align}
		&\Big( \int e^{-\frac{(y-x)^2}{2\sigma^2}} \dif \pi(x,x') \Big)^\alpha \nonumber \\
		&= \Big( \int e^{-\frac{(y-x')^2 + 2(y-x')(x'-x) + (x-x')^2}{2\sigma^2}} \dif \pi(x,x') \Big)^\alpha \nonumber \\
		&= \Big( \int e^{-\frac{(y-x')^2}{2\sigma^2} (\frac{1}{\alpha} + \frac{\alpha-1}{\alpha})}  e^{-\frac{ 2(y-x')(x'-x) + (x-x')^2}{2\sigma^2}} \dif \pi(x,x') \Big)^\alpha  \nonumber \\
		& \leq \Big( \int e^{-\frac{(y-x')^2}{2\sigma^2}} \dif \pi(x,x')  \Big)^{\alpha-1}  \cdot \nonumber \\
		&\qquad \int e^{-\frac{(y-x')^2+2\alpha(y-x')(x'-x) + \alpha (x-x')^2}{2\sigma^2}} \dif \pi(x,x'). \label{eq:HolderInq_Gauss}
	\end{align}
	Here, \eqref{eq:HolderInq_Gauss} is obtained by applying the H\"{o}lder's inequality. Then, 
	\begin{align}
		&D_{\alpha}  (P_{Y|s_i,\rho} \|  P_{Y|s_j,\rho}) \nonumber \\ 
		&\leq \frac{1}{\alpha-1} \log \iint e^{-\frac{(y-x')^2+2\alpha(y-x')(x'-x) + \alpha (x-x')^2}{2\sigma^2}} \dif \pi(x,x') \dif y \nonumber \\
		&=   \frac{1}{\alpha-1} \log \iint \frac{1}{\sqrt{2\pi} \sigma}
		e^{-\frac{(y-x')^2+2\alpha(y-x')(x'-x) + \alpha^2 (x-x')^2}{2\sigma^2}}  \cdot \nonumber \\
		&\qquad\qquad\qquad\qquad\qquad   e^{\frac{  \alpha^2 (x-x')^2 - \alpha (x-x')^2}{2\sigma^2}}  \dif \pi(x,x')  \dif y \nonumber \\
		&=  \frac{1}{\alpha-1} \log \int \Big( \int P_{N}(y + (\alpha - 1) x'-\alpha x)) \dif y \Big) \cdot \nonumber \\
		&\qquad\qquad\qquad\qquad\qquad   e^{\frac{  \alpha^2 (x-x')^2 - \alpha (x-x')^2}{2\sigma^2}}  \dif \pi(x,x') \nonumber \\
		&=  \frac{1}{\alpha-1} \log \int  e^{  \alpha (\alpha-1)  \frac{ (x-x')^2}{2\sigma^2}}  \dif \pi(x,x').  \label{eq:AlphaW_Gauss_Raw}
	\end{align}
	It suffices to request \eqref{eq:AlphaW_Gauss_Raw} no greater than $\epsilon$, and taking the infimum over $\pi$ of the integral, we have
	$$ \inf_{\pi} \int  e^{  \alpha (\alpha-1)  \frac{ (x-x')^2}{2\sigma^2}}  \dif \pi(x,x')  \leq e^{(\alpha-1)\epsilon}, $$
	where the LHS is a $W_1$-metric for convex $d(z) = e^{  \alpha (\alpha-1)  \frac{ z^2}{2\sigma^2}} $ for all $z \in \Real$. We substitute the Kantorovich optimal transport plan $\pi^*$ and get \eqref{eq:theorem:AlphaW_Gaussian}. 
	Requesting this sufficient condition to be satisfied for all $\rho$ and secret pair $(s_i,s_j)$, Theorem~\ref{theorem:AlphaW_Gauss} is proved. 
\end{IEEEproof}

\section{A Sufficient Condition for $\alpha \in (0,1)$} \label{app:Cond_AlpahZeroOne}

\begin{proposition} \label{prop:Cond_AlpahZeroOne}
	For $\alpha \in (0,1)$, $D_\alpha (P_{Y|s_i,\rho}, P_{Y|s_j,\rho}) \leq \epsilon$, if 
	$$ \int e^{-\alpha\frac{|x-x'|}{b}}  \dif \pi(x,x') \geq  e^{(\alpha-1) \epsilon} . $$
\end{proposition}
\begin{IEEEproof}
	Recall the expression of $D_{\alpha} (P_{Y|s_i,\rho} \| P_{Y|s_j,\rho})$ for Laplace noise in \eqref{eq:AllPi} for all $\alpha \in (0,\infty)$. We apply the triangular inequality and H\"{o}lder's inequality as follows.  	
	\begin{align*}
		&\Big( \int e^{-\frac{|y-x|}{b}} \dif \pi(x,x') \Big)^\alpha \\
		& \geq  \Big( \int e^{-\frac{|y-x'|}{b}} e^{-\frac{|x-x'|}{b}} \dif \pi(x,x') \Big)^\alpha\\
		& =  \Big( \int e^{-\frac{|y-x'|}{b} \frac{1}{\alpha}} e^{-\frac{|y-x'|}{b} \frac{\alpha-1}{\alpha}}  e^{-\frac{|x-x'|}{b}} \dif \pi(x,x') \Big)^\alpha \\
		& \geq \Big( \int e^{-\frac{|y-x'|}{b}} \dif \pi(x,x')  \Big)^{\alpha-1}  \int e^{-\frac{|y-x'|}{b}}  e^{\alpha \frac{|x-x'|}{b}} \dif \pi(x,x'), 
	\end{align*}	
	for all $\alpha \in (0,1)$. Here, we still adopt H\"{o}lder conjugates $\alpha$ and $\frac{\alpha}{\alpha-1}$. But, the inequality is reversed as $\frac{\alpha}{\alpha-1}$ is negative~\cite{HolderNegExp2017}. 
	Then, for $\alpha \in (0,1)$,
	\begin{align*}
		D_{\alpha} &(P_{Y|s_i,\rho} \| P_{Y|s_j,\rho})\\
		&=   \frac{1}{\alpha-1} \log \int \frac{1}{2b}
		\frac{  \big( \int e^{-\frac{|y-x|}{b}} \dif \pi(x,x') \big)^\alpha }{ \big( \int e^{-\frac{|y-x'|}{b}} \dif \pi(x,x')  \big)^{\alpha-1}  } \dif y \\
		&\leq  \frac{1}{\alpha-1} \log \int \frac{1}{2b}  \int e^{-\frac{|y-x'|}{b}}  e^{\alpha \frac{-|x-x'|}{b}} \dif \pi(x,x') \dif y  \\
		&=  \frac{1}{\alpha-1} \log \int \Big( \int \frac{1}{2b} e^{-\frac{|y-x'|}{b}}  \dif y \Big) e^{\alpha \frac{-|x-x'|}{b}} \dif \pi(x,x') \\
		&= \frac{1}{\alpha-1} \log  \int e^{\alpha \frac{-|x-x'|}{b}} \dif \pi(x,x'). 
	\end{align*}
	Request it to be upper bounded by $\epsilon$, proposition is proved. 
\end{IEEEproof}

\section{More Experimental Results}
\label{sec:ExpExtra}

Figure~\ref{fig:app1} shows how the LHS $ W_{\alpha}(P_{X|s_i,\rho}, P_{X|s_j,\rho}) $and RHS $e^{ \frac{\alpha-1}{\alpha} \epsilon }$ of the inequality~\eqref{eq:AlphaW_Lap_Min_Alt} varies with $\alpha$, for different values of scale parameter $b$ of Laplace noise. 
Figure~\ref{fig:app2} is an example of Remark~\ref{rem:Laplace}. It shows the scale parameter $b$ determined by Theorem~\ref{theorem:AlphaW_Laplace} and \cite[Corollary~3.1]{RPP2024} both converges to $W_\infty/\epsilon$, the $\infty$-Wasserstein mechanism in \cite{PufferfishWasserstein2017Song}, when $\alpha$ grows large.

Figure~\ref{fig:app1} also shows for smaller value of $\alpha$, a larger scale parameter $b$ for Laplace noise should be chosen to satisfy the sufficient condition in Theorem~\ref{theorem:AlphaW_Laplace}.

\bibliographystyle{IEEEtran}
\bibliography{BIB}

@InProceedings{CalibNoiseDP,
  author    = {Dwork, Cynthia and McSherry, Frank and Nissim, Kobbi and Smith, Adam},
  booktitle = {Theory of Cryptography},
  title     = {Calibrating Noise to Sensitivity in Private Data Analysis},
  year      = {2006},
  address   = {Berlin, Heidelberg},
  editor    = {Halevi, Shai and Rabin, Tal},
  pages     = {265--284},
  publisher = {Springer Berlin Heidelberg},
  isbn      = {978-3-540-32732-5},
}

@InProceedings{RDP2017,
  author    = {I. {Mironov}},
  booktitle = {2017 IEEE 30th Computer Security Foundations Symposium (CSF)},
  title     = {Rényi Differential Privacy},
  year      = {2017},
  pages     = {263-275},
}

@Article{Pufferfish2014Kifer,
  author     = {Kifer, Daniel and Machanavajjhala, Ashwin},
  title      = {Pufferfish: A Framework for Mathematical Privacy Definitions},
  journal    = {ACM Transactions on Database Systems},
  year       = {2014},
  volume     = {39},
  number     = {1},
  month      = jan,
  address    = {New York, NY, USA},
  articleno  = {3},
  doi        = {10.1145/2514689},
  issn       = {0362-5915},
  issue_date = {January 2014},
  numpages   = {36},
  publisher  = {Association for Computing Machinery},
}

@InProceedings{PufferfishWasserstein2017Song,
  author    = {Song, Shuang and Wang, Yizhen and Chaudhuri, Kamalika},
  booktitle = {Proceedings of the 2017 ACM International Conference on Management of Data},
  title     = {Pufferfish Privacy Mechanisms for Correlated Data},
  year      = {2017},
  address   = {New York, NY, USA},
  pages     = {1291–1306},
  location  = {Chicago, Illinois, USA},
  numpages  = {16},
}

@article{Dwork2014book,
  title={The algorithmic foundations of differential privacy.},
  author={Dwork, Cynthia and Roth, Aaron and others},
  journal={Found. Trends Theor. Comput. Sci.},
  volume={9},
  number={3-4},
  pages={211--407},
  year={2014}
}

@Misc{UCI2007,
  author = {Asuncion, Arthur and Newman, David},
  title  = {{UCI} machine learning repository {https://archive.ics.uci.edu/ml/index.php}},
  year   = {2007},
  url    = {https://archive.ics.uci.edu/ml/index.php},
}

@Book{Villani2009OPT,
  title     = {Optimal transport: old and new},
  publisher = {Springer},
  year      = {2009},
  author    = {Villani, C{\'e}dric},
  volume    = {338},
}

@Article{Santambrogio2015OPT,
  author    = {Santambrogio, Filippo},
  title     = {Optimal transport for applied mathematicians},
  journal   = {Birk{\"a}user, NY},
  year      = {2015},
  volume    = {55},
  number    = {58-63},
  pages     = {94},
  publisher = {Springer},
}

@InProceedings{Dwork2006,
  author    = {Dwork, Cynthia},
  title     = {Differential Privacy},
  booktitle = {Automata, Languages and Programming},
  year      = {2006},
  editor    = {Bugliesi, Michele and Preneel, Bart and Sassone, Vladimiro and Wegener, Ingo},
  pages     = {1--12},
  address   = {Berlin, Heidelberg},
  publisher = {Springer Berlin Heidelberg},
}

@InProceedings{Pufferfish2012KiferConf,
  author    = {Kifer, Daniel and Machanavajjhala, Ashwin},
  title     = {A Rigorous and Customizable Framework for Privacy},
  booktitle = {Proceedings of the 31st ACM SIGMOD-SIGACT-SIGAI Symposium on Principles of Database Systems},
  year      = {2012},
  series    = {PODS '12},
  pages     = {77–88},
  address   = {New York, NY, USA},
  publisher = {Association for Computing Machinery},
  isbn      = {9781450312486},
  location  = {Scottsdale, Arizona, USA},
  numpages  = {12},
}

@Article{Champion2008InfW,
  author    = {Champion, Thierry and De Pascale, Luigi and Juutinen, Petri},
  title     = {The $\infty$-{W}asserstein Distance: Local Solutions and Existence of Optimal Transport Maps},
  journal   = {SIAM Journal on Mathematical Analysis},
  year      = {2008},
  volume    = {40},
  number    = {1},
  pages     = {1--20},
  publisher = {SIAM},
}

@Article{DePascale2019InfW,
  author  = {Luigi {De Pascale} and Jean Louet},
  title   = {A study of the dual problem of the one-dimensional $L_{\infty}$-optimal transport problem with applications},
  journal = {Journal of Functional Analysis},
  year    = {2019},
  volume  = {276},
  number  = {11},
  pages   = {3304-3324},
  issn    = {0022-1236},
}

@InProceedings{Ding2022AISTATS,
  author    = {Ding, Ni},
  booktitle = {Proceedings of The 25th International Conference on Artificial Intelligence and Statistics},
  title     = {Kantorovich Mechanism for Pufferfish Privacy},
  year      = {2022},
  editor    = {Camps-Valls, Gustau and Ruiz, Francisco J. R. and Valera, Isabel},
  month     = {28--30 Mar},
  pages     = {5084--5103},
  publisher = {PMLR},
  series    = {Proceedings of Machine Learning Research},
  volume    = {151},
  url       = {https://proceedings.mlr.press/v151/ding22b.html},
}

@Article{Wasserman2010,
  author    = {Wasserman, Larry and Zhou, Shuheng},
  journal   = {Journal of the American Statistical Association},
  title     = {A Statistical Framework for Differential Privacy},
  year      = {2010},
  issn      = {1537-274X},
  month     = mar,
  number    = {489},
  pages     = {375--389},
  volume    = {105},
  doi       = {10.1198/jasa.2009.tm08651},
  publisher = {Informa UK Limited},
}

@Article{Liao2019_AlphaLeak,
  author  = {Liao, Jiachun and Kosut, Oliver and Sankar, Lalitha and du Pin Calmon, Flavio},
  journal = {IEEE Transactions on Information Theory},
  title   = {Tunable Measures for Information Leakage and Applications to Privacy-Utility Tradeoffs},
  year    = {2019},
  month   = dec,
  number  = {12},
  pages   = {8043-8066},
  volume  = {65},
  doi     = {10.1109/tit.2019.2935768},
}

@InProceedings{RPP2024,
  author      = {Pierquin, Cl{\'e}ment and Bellet, Aur{\'e}lien and Tommasi, Marc and Boussard, Matthieu},
  booktitle   = {{International Conference on Machine Learning (ICML 2024)}},
  title       = {{R{\'e}nyi Pufferfish Privacy: General Additive Noise Mechanisms and Privacy Amplification by Iteration via Shift Reduction Lemmas}},
  year        = {2024},
  address     = {Vienna (Austria), Austria},
  hal_id      = {hal-04363020},
  hal_version = {v2},
  url         = {https://inria.hal.science/hal-04363020},
}

@Article{Brent1971,
  author    = {Brent, R. P.},
  journal   = {The Computer Journal},
  title     = {An algorithm with guaranteed convergence for finding a zero of a function},
  year      = {1971},
  issn      = {1460-2067},
  month     = apr,
  number    = {4},
  pages     = {422--425},
  volume    = {14},
  doi       = {10.1093/comjnl/14.4.422},
  publisher = {Oxford University Press (OUP)},
}

@Book{Suli2003book,
  author    = {S{\"u}li, Endre and Mayers, David F},
  publisher = {Cambridge university press},
  title     = {An introduction to numerical analysis},
  year      = {2003},
}

@Article{Erven2014_JOURNAL,
  author  = {Tim van Erven and Peter Harremoes},
  journal = {{IEEE} Transactions on Information Theory},
  title   = {{R{\'{e}}nyi} Divergence and {Kullback-Leibler} Divergence},
  year    = {2014},
  month   = jul,
  number  = {7},
  pages   = {3797-3820},
  volume  = {60},
  doi     = {10.1109/tit.2014.2320500},
}

@Article{Esposito2020_CONF,
  author    = {Amedeo Roberto Esposito and Michael Gastpar and Ibrahim Issa},
  title     = {Robust Generalization via $f$-Mutual Information},
  year      = {2020},
  pages     = {2723-2728},
  address   = {Los Angeles, CA},
  booktitle = {Proceedings of {IEEE} International Symposium on Information Theory},
  doi       = {10.1109/isit44484.2020.9174117},
}

@Article{Esposito2021_JOURNAL,
  author  = {Amedeo Roberto Esposito and Michael Gastpar and Ibrahim Issa},
  journal = {{IEEE} Transactions on Information Theory},
  title   = {Generalization Error Bounds via {R{\'{e}}nyi}-, $f$-Divergences and Maximal Leakage},
  year    = {2021},
  month   = aug,
  number  = {8},
  pages   = {4986-5004},
  volume  = {67},
  doi     = {10.1109/tit.2021.3085190},
}

@InProceedings{Renyi1961_Measures,
  author       = {R{\'e}nyi, Alfr{\'e}d},
  booktitle    = {Proceedings of the Fourth Berkeley Symposium on Mathematical Statistics and Probability, Volume 1: Contributions to the Theory of Statistics},
  title        = {On measures of entropy and information},
  year         = {1961},
  organization = {University of California Press},
  pages        = {547-562},
  volume       = {4},
}

@InProceedings{Massey1994Guess,
  author     = {Massey, J.L.},
  booktitle  = {Proceedings of 1994 IEEE International Symposium on Information Theory},
  title      = {Guessing and entropy},
  pages      = {204},
  publisher  = {IEEE},
  series     = {ISIT-94},
  collection = {ISIT-94},
  doi        = {10.1109/isit.1994.394764},
}

@Article{Arikan1996Guess,
  author    = {Arikan, E.},
  journal   = {IEEE Transactions on Information Theory},
  title     = {An inequality on guessing and its application to sequential decoding},
  year      = {1996},
  issn      = {0018-9448},
  number    = {1},
  pages     = {99--105},
  volume    = {42},
  doi       = {10.1109/18.481781},
  publisher = {Institute of Electrical and Electronics Engineers (IEEE)},
}

@InProceedings{Rioul2020ISITA,
  author    = {Rioul, Olivier},
  booktitle = {2020 International Symposium on Information Theory and Its Applications (ISITA)},
  title     = {Rényi Entropy Power and Normal Transport},
  year      = {2020},
  pages     = {1-5},
}

@InProceedings{Feldman2018FOCS,
  author    = {Feldman, Vitaly and Mironov, Ilya and Talwar, Kunal and Thakurta, Abhradeep},
  booktitle = {2018 IEEE 59th Annual Symposium on Foundations of Computer Science (FOCS)},
  title     = {Privacy Amplification by Iteration},
  year      = {2018},
  month     = Oct,
  pages     = {521--532},
  publisher = {IEEE},
  doi       = {10.1109/focs.2018.00056},
}

@Article{Rioul2011,
  author    = {Rioul, Olivier},
  journal   = {IEEE Transactions on Information Theory},
  title     = {Information Theoretic Proofs of Entropy Power Inequalities},
  year      = {2011},
  issn      = {1557-9654},
  month     = Jan,
  number    = {1},
  pages     = {33--55},
  volume    = {57},
  doi       = {10.1109/tit.2010.2090193},
  publisher = {Institute of Electrical and Electronics Engineers (IEEE)},
}

@InProceedings{Ding2025ITW,
  author    = {Ding, Ni and Farokhi, Farhad and Guo, Tao and Xu, Yinfei and Zhang, Xiang},
  booktitle = {2025 IEEE Information Theory Workshop (ITW)},
  title     = {$\alpha$-leakage Interpretation of Sibson Mutual Information and R\'{e}nyi Capacity},
  year      = {2025},
  month     = Sept,
  pages     = {752--757},
  publisher = {IEEE},
  doi       = {10.1109/itw62417.2025.11240362},
}

@Article{Dowson1982Frechet,
  author    = {Dowson, DC and Landau, BV666017},
  journal   = {Journal of multivariate analysis},
  title     = {The Fr{\'e}chet distance between multivariate normal distributions},
  year      = {1982},
  number    = {3},
  pages     = {450--455},
  volume    = {12},
  publisher = {Elsevier},
}

@Article{Givens1984NormClass,
  author    = {Givens, Clark R and Shortt, Rae Michael},
  journal   = {Michigan Mathematical Journal},
  title     = {A class of {W}asserstein metrics for probability distributions.},
  year      = {1984},
  number    = {2},
  pages     = {231--240},
  volume    = {31},
  publisher = {University of Michigan, Department of Mathematics},
}

@Article{Takatsu2011Wnorm,
  author    = {Takatsu, Asuka},
  journal   = {Osaka Journal of Mathematics},
  title     = {{W}asserstein geometry of {G}aussian measures},
  year      = {2011},
  number    = {4},
  pages     = {1005--1026},
  volume    = {48},
  publisher = {Osaka University},
}

@Article{HolderNegExp2017,
  author    = {Daoxiang, Zhang and Yan, Ping},
  journal   = {Journal of mathematical inequalities},
  title     = {ON THE HARDY--CARLEMAN INEQUALITY FOR A NEGATIVE EXPONENT},
  year      = {2017},
  number    = {3},
  pages     = {885--890},
  volume    = {11},
  publisher = {ELEMENT R AUSTRIJE 11, 10000 ZAGREB, CROATIA},
}

@InProceedings{Abowd2018USCensus,
  author     = {Abowd, John M.},
  booktitle  = {Proceedings of the 24th ACM SIGKDD International Conference on Knowledge Discovery and Data Mining},
  title      = {The U.S. Census Bureau Adopts Differential Privacy},
  year       = {2018},
  month      = July,
  pages      = {2867--2867},
  publisher  = {ACM},
  series     = {KDD ’18},
  collection = {KDD ’18},
  doi        = {10.1145/3219819.3226070},
}

@InProceedings{Abadi2016DPSGD,
  author     = {Abadi, Martin and Chu, Andy and Goodfellow, Ian and McMahan, H. Brendan and Mironov, Ilya and Talwar, Kunal and Zhang, Li},
  booktitle  = {Proceedings of the 2016 ACM SIGSAC Conference on Computer and Communications Security},
  title      = {Deep Learning with Differential Privacy},
  year       = {2016},
  month      = Oct,
  pages      = {308--318},
  publisher  = {ACM},
  collection = {CCS′16},
  doi        = {10.1145/2976749.2978318},
}

@Article{Mohammadi2025Medicine,
  author  = {Mohammadi, Marziyeh and Vejdanihemmat, Mohsen and Lotfinia, Mahshad and Rusu, Mirabela and Truhn, Daniel and Maier, Andreas and Tayebi Arasteh, Soroosh},
  journal = {arXiv e-prints},
  title   = {Differential privacy for deep learning in medicine},
  year    = {2025},
  pages   = {arXiv--2506},
}

@Article{SoriaComas2017,
  author    = {Soria-Comas, Jordi and Domingo-Ferrer, Josep and Sanchez, David and Megias, David},
  journal   = {IEEE Transactions on Information Forensics and Security},
  title     = {Individual Differential Privacy: A Utility-Preserving Formulation of Differential Privacy Guarantees},
  year      = {2017},
  issn      = {1556-6021},
  month     = June,
  number    = {6},
  pages     = {1418--1429},
  volume    = {12},
  doi       = {10.1109/tifs.2017.2663337},
  publisher = {Institute of Electrical and Electronics Engineers (IEEE)},
}

@InProceedings{Li2024,
  author     = {Li, Bo and Wang, Wei and Ye, Peng},
  booktitle  = {Advances in Neural Information Processing Systems 37},
  title      = {The Limits of Differential Privacy in Online Learning},
  year       = {2024},
  pages      = {65328--65360},
  publisher  = {Neural Information Processing Systems Foundation, Inc. (NeurIPS)},
  series     = {NeurIPS 2024},
  collection = {NeurIPS 2024},
  doi        = {10.52202/079017-2086},
}

@Article{Yang2026Noise,
  author  = {Yang, Wenjin and Ding, Ni and Zhang, Zijian and Sun, Jing and Li, Zhen and Wu, Yan and Sun, Jiahang and Lin, Haotian and Liu, Yong and An, Jincheng and others},
  journal = {arXiv preprint arXiv:2601.06385},
  title   = {Noise Reduction for Pufferfish Privacy: A Practical Noise Calibration Method},
  year    = {2026},
}

@Article{Ding2025Multi,
  author  = {Ding, Ni and Lu, Songpei and Yang, Wenjing and Zhang, Zijian},
  journal = {arXiv preprint arXiv:2512.18632},
  title   = {Multi-user Pufferfish Privacy},
  year    = {2025},
}

@InProceedings{Balle2018AnalyticGM,
  author    = {Balle, Borja and Wang, Yu-Xiang},
  booktitle = {Proceedings of the 35th International Conference on Machine Learning},
  title     = {Improving the {G}aussian Mechanism for Differential Privacy: Analytical Calibration and Optimal Denoising},
  year      = {2018},
  editor    = {Dy, Jennifer and Krause, Andreas},
  month     = {10--15 Jul},
  pages     = {394--403},
  publisher = {PMLR},
  series    = {Proceedings of Machine Learning Research},
  volume    = {80},
}

@Article{Mironov2019RDP,
  author  = {Mironov, Ilya and Talwar, Kunal and Zhang, Li},
  journal = {arXiv preprint arXiv:1908.10530},
  title   = {R\'{e}nyi differential privacy of the sampled gaussian mechanism},
  year    = {2019},
}

@Article{Thorpe2018OPTIntro,
  author  = {Thorpe, Matthew},
  journal = {Notes of Course at University of Cambridge},
  title   = {Introduction to optimal transport},
  year    = {2018},
  volume  = {3},
}

@InProceedings{Sason2015,
  author    = {Sason, Igal and Verdu, Sergio},
  booktitle = {2015 IEEE Information Theory Workshop - Fall (ITW)},
  title     = {Upper bounds on the relative entropy and Rényi divergence as a function of total variation distance for finite alphabets},
  year      = {2015},
  month     = Oct,
  pages     = {214--218},
  publisher = {IEEE},
  doi       = {10.1109/itwf.2015.7360766},
}

@Article{Gilardoni2010,
  author    = {Gilardoni, Gustavo L.},
  journal   = {IEEE Transactions on Information Theory},
  title     = {On Pinsker’s and Vajda’s Type Inequalities for Csiszár’s $f$-Divergences},
  year      = {2010},
  issn      = {0018-9448},
  month     = Nov,
  number    = {11},
  pages     = {5377--5386},
  volume    = {56},
  doi       = {10.1109/tit.2010.2068710},
  publisher = {Institute of Electrical and Electronics Engineers (IEEE)},
}

@InProceedings{Asoodeh2021,
  author    = {Asoodeh, Shahab and Aliakbarpour, Maryam and Calmon, Flavio P.},
  booktitle = {2021 IEEE International Symposium on Information Theory (ISIT)},
  title     = {Local Differential Privacy Is Equivalent to Contraction of an $f$-Divergence},
  year      = {2021},
  month     = July,
  pages     = {545--550},
  publisher = {IEEE},
  doi       = {10.1109/isit45174.2021.9517999},
}

@Article{Sason2016,
  author    = {Sason, Igal and Verdu, Sergio},
  journal   = {IEEE Transactions on Information Theory},
  title     = {$f$ -Divergence Inequalities},
  year      = {2016},
  issn      = {1557-9654},
  month     = Nov,
  number    = {11},
  pages     = {5973--6006},
  volume    = {62},
  doi       = {10.1109/tit.2016.2603151},
  publisher = {Institute of Electrical and Electronics Engineers (IEEE)},
}

%
%

\end{document}